\documentclass[10pt,twocolumn]{article}

\usepackage[margin=0.75in]{geometry}
\usepackage[utf8x]{inputenc}
\usepackage[T1]{fontenc}
\usepackage{times}
\usepackage{cite}
\usepackage{graphics} 
\usepackage{epsfig} 
\usepackage{times} 
\usepackage{subcaption}
\usepackage{enumerate,url,tikz,adjustbox}
\usepackage{balance}
\usepackage[toc,page]{appendix}
\usepackage{pgfplots}

\pgfplotsset{compat=1.18}

\usepackage{titlesec}

\titleformat{\section}
  {\normalfont\large\bfseries}{\thesection.}{0.5em}{}
  
\titleformat{\subsection}
  {\normalfont\normalsize\itshape}{\thesubsection.}{1em}{}
  
\usepackage{amssymb,amsfonts,amscd,dsfont,nicefrac,amsmath,amsthm} 
\usepackage{xcolor}

\usetikzlibrary{calc,arrows}

\usepackage[ruled,linesnumbered,lined,vlined]{algorithm2e}

\renewcommand{\Re}{\mathbb{R}}
\newcommand{\tp}{{\mbox{\tiny \sf T}}}
\newcommand{\Nint}{\mathcal N_{\text{int}}}
\newcommand{\Nlf}{\mathcal N_{\text{leaf}}}
\newcommand{\T}{\mathcal T}
\newcommand{\N}{\mathcal N}
\newcommand{\chd}{\mathcal C}
\newcommand{\W}{\mathcal W}
\newcommand{\Q}{\mathcal Q}
\newcommand{\z}{\mathbf z}
\DeclareMathOperator*{\argmax}{argmax}
\DeclareMathOperator*{\argmin}{argmin}
\newcommand{\tRoot}{t_{\mathsf R}}

\newcommand{\DKL}{\mathrm{D}_{\mathrm{KL}}}
\newcommand{\JS}{\mathrm{JS}_{\Pi}}

\renewcommand{\vec}[1]{\mathbf{#1}}
\newcommand{\betacrq}{\beta^{Q}_{\text{cr}}}
\newcommand{\betacr}{\beta_{\text{cr}}}
\newcommand{\betaPhaseTransitionSet}{\mathcal B_{\text{PT}}}
\newcommand{\betacrqset}{\mathcal B_{\text{cr}}^Q}
\newcommand{\betaPhaseTransitionSetNode}[1]{\mathcal B_{\text{PT},{#1}}}

\newcommand{\figref}[1]{Figure~\ref{#1}}

\usepackage{enumitem}

\newtheorem{theorem}{Theorem}[section]
\newtheorem{lemma}[theorem]{Lemma}

\newtheorem{proposition}[theorem]{Proposition}

\newtheorem{definition}[theorem]{Definition}

\renewcommand\footnotemark{}

\renewcommand{\thesection}{\arabic{section}} 
\renewcommand{\thesubsection}{\Alph{subsection}}

\title{\LARGE \bf A Dual Approach for Hierarchical Information-Theoretic\\ Tree Abstractions}
\date{}
\author{Daniel T. Larsson$^{1}$ \and Dipankar Maity$^{2}$ \and Panagiotis Tsiotras$^{3}$
\footnote{This research was funded by Office of Naval Research award N00014-18-1-2375 and by the Army Research Laboratory under DCIST CRA W911NF-17-2-0181.}
\footnote{$^{1}$Assistant Professor in the Aerospace and Mechanical Engineering Department at the University of Arizona, Tucson, AZ, 85721, USA. {\tt\small dlarsson@arizona.edu}}%
\footnote{$^{2}$Assistant Professor in the Department of Electrical and Computer Engineering at the University of North Carolina at Charlotte, NC, 28223, USA. {\tt\small dmaity@uncc.edu}}%
\footnote{$^{3}$Andrew and Lewis Chair Professor with the D. Guggenheim School of Aerospace Engineering and the Institute for Robotics and Intelligent Machines, Georgia Institute of Technology, Atlanta, GA, 30332-0150, USA. {\tt\small tsiotras@gatech.edu}}%
}

\begin{document} 

\maketitle

\begin{abstract}
In this paper, we consider establishing a formal connection between two distinct tree-abstraction problems inspired by the information-bottleneck (IB) method.
Specifically, we consider the hard- and soft-constrained formulations that have recently appeared in the literature to determine the conditions for which the two approaches are equivalent.
%
Our analysis leverages concepts from Lagrangian relaxation and duality theory to relate the dual function of the hard-constrained problem to the Q-function employed in Q-tree search and shows the connection between tree phase transitions and solutions to the dual problem obtained by exploiting the problem structure.
An algorithm is proposed that employs knowledge of the tree phase transitions to find a setting of the dual variable that solves the dual problem.
Furthermore, we present an alternative approach to select the dual variable that leverages the integer programming formulation of the hard-constrained problem and the strong duality of linear programming.
To obtain a linear program, we establish that a relaxation of the integer programming formulation of the hard-constrained tree-search problem has the integrality property by showing that the program constraint matrix is totally unimodular.
Empirical results that corroborate the theoretical developments are presented and discussed throughout.
\end{abstract}

\section{Introduction}

The ability to distill streams of information to identify the most salient and relevant features, and removing, or discarding, those aspects of the signal that are unimportant, or irrelevant, (i.e., abstraction) is considered a key attribute of intelligent reasoning~\cite{holte2003abstraction,ponsen2010abstraction,zucker2003grounded}.
As a result, researchers within the artificial intelligence and autonomous systems communities have considered the development of a number of frameworks that both consider the design and deployment of abstractions in systems applications.
However, despite the introduction of many frameworks that leverage abstraction in various aspects of autonomy, few approaches have emerged through the years that consider the principled development and design of the abstractions themselves. 
To date, the use of abstractions in autonomous systems is mainly motivated and considered from the perspective of developing systems that can tailor their (on-board) processing in accordance with system resource constraints.
For example, in the artificial intelligence community, the works of~\cite{holte2003abstraction,li2006towards} consider the development of abstractions for planning, decision-making, and problem-solving from a systematic perspective, providing some insights into the foundations for what defines an abstraction, especially with respect to their use in reducing the difficulty of obtaining solutions to challenging problems.
In other work, such as~\cite{hauer2015multi,kambhampati1986multiresolution,behnke2004local,cowlagi2008multiresolution}, researchers within the autonomous planning community have considered graph abstractions (i.e., compression or reduction) in order to reduce the number of vertices for which graph-search algorithms such as Dijkstra or \(A^*\) must visit, thereby reducing the computational burden of planning.
In a similar line of research, the authors of~\cite{larsson2022ucsdPerception} consider the design of abstractions in environments containing semantic information (e.g., trees, grass, asphalt), and show how the resulting abstractions, in the form of multi-resolution octrees, may be utilized in order to create reduced (compressed) colored graphs to speed semantic-constrained planning.
The use of abstractions in stochastic decision-making problems has also been well-documented~\cite{botvinick2012hierarchical,tishby2010information,genewein2015bounded,ponsen2010abstraction}.
Of specific interest is~\cite{tishby2010information,genewein2015bounded}, where researchers within the bounded rationality community consider the development of frameworks for single-step and sequential decision-making in agents that are subject to information-processing constraints.
In the aforementioned works, the authors employ ideas from information-theory to model information-limited agents as a reduced-capacity channel between states and actions, and derive the analytical structure of optimal decision rules in information-constrained Markov decision processes (MDPs) and maximum expected utility problems.
Abstractions arise in the works~\cite{tishby2010information,genewein2015bounded} as stochastic policies that become less and less state specific as the agent's information processing capabilities become more and more limited, thereby reducing the specificity to which the agent must localize itself (i.e., determine which state it is in) when deciding how to act. 
Information-theoretically inspired frameworks to model information-limited decision-making have also been extended to multi-agent systems~\cite{guan2022hierarchical,grau2018balancing}.
Few frameworks have been developed that specifically focus on the principled design of abstract representations that are amenable to autonomous planning, perception, and control.
To address this shortcoming, the work by the authors of~\cite{larsson2020q,larsson2021information,larsson2022generalized,larsson2022linearBC,larsson2021informationB} considers the development of information-theoretically inspired frameworks for the emergence of hierarchical, multi-resolution, abstractions in autonomous systems.
More specifically, the work of~\cite{larsson2020q,larsson2021information,larsson2022generalized,larsson2022linearBC} leverages ideas from the information-bottleneck (IB) problem~\cite{Tishby1999} and its variants~\cite{chechik2002extracting}, and draws connections between hierarchical, multi-scale, tree data structures and signal encoders to develop a new framework to generate abstractions for resource-limited autonomous systems that are emergent from an information-theoretic principle that is rooted in signal compression theory.
Interestingly, the problem introduced in~\cite{larsson2020q} is a soft-constrained, penalty-function-like, approach that likens the original information-bottleneck problem, whereby the balance between information-retention and compression in the resulting multi-resolution abstraction is specified by a trade-off parameter.
In other related work, the authors of~\cite{larsson2022linearBC,larsson2021information} introduce a hard-constrained version of the information-theoretic abstraction problem, where the balance between compression and information-retention is not specified in terms of a trade-off parameter, but rather enforced as a hard-constraint that alters the feasible set of the problem, thereby guaranteeing that the resulting abstraction retains a desired level of task-relevant information (or equivalently, achieve a desired level of compression).
The authors show that a solution to the hard-constrained problem may be obtained by solving an integer linear program, whereas a solution to the soft-constrained version of the information-theoretic tree abstraction problem~\cite{larsson2020q} is obtained by means of an algorithm called Q-tree search, which exploits the tree structure to tractably find an optimal, multi-resolution, tree abstraction.
However, while these two approaches have been shown to successfully generate abstractions for perception and planning in autonomous systems~\cite{larsson2021informationB,larsson2022ucsdPerception}, a clear connection between the two methods is not yet available.
The importance of establishing a connection between the two methods lies in the guarantees furnished by the resulting solutions, which are formulation-dependent.
Specifically, consider, for example, an agent who must design an abstraction for onboard storage or for communicating with other agents across a capacity-limited channel.
In such a setting, the ability to design an abstraction satisfying hard-constraints is clearly useful due to the system constraints.
In contrast, the Q-tree search algorithm is a tractable approach to obtaining multi-resolution abstractions, but its solutions depend on the specification of a trade-off parameter, whose selection as a function of the hard-constraint is not trivial.
Thus, the development of a framework that allows the trade-off parameter in Q-tree search to be selected as a function of system-level specifications represented by hard constraints is of great importance.
Consequently, the contributions of this paper are as follows.
First, we develop a connection between the hard-constrained, integer programming~\cite{larsson2022linearBC,larsson2021information} and the soft-constrained, penalty-function-like~\cite{larsson2020q} approaches to the information-theoretic hierarchical tree abstraction problem.
To accomplish our goal, we employ concepts from relaxation and duality theory and ultimately show that the dual function of the hard-constrained problem may be written in terms of the Q-function from Q-tree search.
Then, having established a formal bridge between the two problems that link Q-tree search and the hard-constrained formulation, we propose an algorithm that solves the resulting dual problem, thereby providing a means by which a setting of the trade-off parameter in Q-tree search may be obtained as a function of the hard-constraint in the integer programming formulation.
Our algorithm leverages tree phase transitions and delineates how they are related to solutions of the dual problem.
Lastly, we show that Q-tree search may be realized as a linear program by considering a suitable relaxation of the hard-constrained integer linear program.
To achieve our goal, we show that the hierarchical tree constraint results in a matrix that is totally unimodular, thereby providing another means by which the trade-off parameter in the soft-constrained problem may be selected as a function of the hard-constraints.
In summary, this paper details the precise connection between the hard- and soft-constrained information-theoretic tree search problems that appeared originally in~\cite{larsson2020q,larsson2021information,larsson2022linearBC}, and develops methods by which the value of the trade-off parameter in the soft-constrained formulation may be selected as a function of the hard constraint without the need to solve a (potentially large) integer program.
%

\section{Preliminaries} \label{sec:prelims}

We begin with comments regarding our notation throughout this paper.
To this end, we take \(\Re\) to be the set of real numbers, and for any integer \(n > 0\) let \(\Re^n\) represent the \(n\)-dimensional Euclidean space; that is \(\Re^n = \{(x_1,\ldots,x_n): x_i\in\Re,~0\leq i\leq n\}\).
For any vector \(\vec{x}\in\Re^n\), we denote the \(i^{\text{th}}\) element of \(\vec{x}\) by \([\vec{x}]_i\) for \(1\leq i \leq n\).
Given any \(\vec{x},\vec{y}\in\Re^n\), the relation \(\vec{x}\leq\vec{y}\) is to be understood component-wise: \([\vec{x}]_i \leq [\vec{y}]_i\) for \(i = 1,\ldots,n\).
Moreover, for integer \(n > 0\), we take \(\Re^n_+\) to be the collection of all \(n\)-dimensional vectors with non-negative components; that is, \(\Re^n_+ = \{\vec{x}\in\Re^n : [\vec{x}]_i \geq 0,~1\leq i\leq n\}\).
Similarly, for integers \(n > 0\) and \(m > 0\), the collection of all \(m\times n\) dimensional real-valued matrices is denoted \(\Re^{m\times n}\), and for any \(A \in\Re^{m\times n}\), the entry in the \(i^{\text{th}}\) row and \(j^{\text{th}}\) column is denoted \([A]_{ij}\) for \(1\leq i \leq m\), and \(1 \leq j \leq n\).
To formulate the information-theoretic compression problems in the sequel, we require the specification of a probability space.
To this end, let \((\Omega,\Sigma,\mu)\) be a probability space with (finite) sample space \(\Omega\), \(\sigma\)-algebra \(\Sigma\) and probability measure \(\mu\).
A random variable \(X: \Omega \to \Re\) is a measurable function with corresponding distribution (mass function) \(p(x) = \mu(\{\omega \in \Omega: X(\omega) = x\})\).
We let the collection of outcomes of the random variable \(X\) we denoted  \(\Omega_X = \{x \in \Re: X(\omega) = x,~\omega\in \Omega\}\).
Given a random variable \(X\), its \emph{Shannon entropy} is defined as \(H(X) = - \sum_{x \in \Omega_X} p(x) \log p(x)\).
The entropy of \(X\) is a measure of the uncertainty in the outcomes of \(X\), and depends only on the distribution \(p(x)\)~\cite{Cover2006}.
Thus, we will at times abuse notation and write \(H(p)\) in place of \(H(X)\).
For any two probability distributions \(p(x)\) and \(q(x)\) over the same collection of outcomes, the \emph{Kullback-Leibler (KL) divergence} is given by \(\DKL(p(x),q(x)) = \sum_{x} p(x) \log ( p(x) / q(x) )\)~\cite{Cover2006}.
Another important divergence measure is that of the \emph{Jensen-Shannon (JS) divergence}~\cite{lin1991divergence}, defined for a collection of distributions \(\{p_1,\ldots,p_n\}\) over the same set of outcomes as 
\begin{equation}\label{eq:JSdivergenceDef}
\JS(p_1,\ldots,p_n) = \sum_{i=1}^n [\Pi]_i \DKL(p_i(x),\bar p(x)),
\end{equation}
where \(\Pi\in\Re_+^n\) are given weights satisfying \(0 \leq [\Pi]_i \leq 1\), \(1\leq i \leq n\), \(\sum_{i=1}^n [\Pi]_i = 1\) and \(\bar p(x) = \sum_{i=1}^{n} [\Pi]_i p_i(x)\).

Consider a random variable \(Y:\Omega \to \Re\) with distribution \(p(y)\).
Then, the \emph{mutual information} between the random variables \(X\) and \(Y\) is defined in terms of the KL-divergence as
\begin{equation}\label{eq:MI}
I(X;Y) = \DKL(p(x,y),p(x)p(y)).
\end{equation}
The mutual information is non-negative, and is a measure of the statistical dependence between the random variables \(X\) and \(Y\)~\cite{Cover2006}.
Importantly, the mutual information may also be written as
\begin{equation}\label{eq:MIconditionalEntropy}
I(X;Y) = H(X) - H(X|Y) = H(Y) - H(Y|X),
\end{equation}
where \(H(X|Y)\) is the conditional entropy of \(X\) given \(Y\) defined by \(H(X|Y) = -\sum_{x,y}p(x,y)\log p(x|y)\).
In this paper, we will consider two problems that employ ideas from the information-bottleneck (IB) method~\cite{Tishby1999}.
The IB method is a signal-compression framework that considers the problem of forming a compressed representation \(T\) of a signal \(X\) so that the compressed representation \(T\) is as informative as possible regarding a third, relevance-variable \(Y\) which contains the relevant information.
The IB problem assumes that the joint distribution is given by \(p(t,x,y) = p(t|x)p(x,y)\), which implies that the random variables \(X\), \(Y\), and \(T\) satisfy the Markov chain conditions \(T \leftarrow X \leftarrow Y\) and \(T \rightarrow X \rightarrow Y\) (oftentimes denoted \(T \leftrightarrow X \leftrightarrow Y\)).
Given the distribution \(p(x,y)\), the IB method seeks to design the compressed representation \(T\) by designing a conditional distribution \(p(t|x)\) so as the solve the problem
\begin{equation}\label{eq:stdIBproblem}
    \min I(T;X) - \beta I(T;Y),
\end{equation}
where \(\beta \geq 0\) is a non-negative parameter that trades the importance of compression (i.e., minimizing \(I(T;X)\)) and information-retention regarding the relevant variable \(Y\) (i.e., maximizing \(I(T;Y)\)).
For more information regarding the IB method, the interested reader is referred to~\cite{Tishby1999,Slonim2002}.
In this paper, we consider developing abstractions of the environment \(\W\) in the form of hierarchical tree structures that emerge as a solution to an information-theoretic signal encoder problem that utilizes ideas from the IB method.
To this end, we assume that there exists an \(\ell > 0\) such that the grid-world environment \(\W \subset \Re^{w}\), is contained within a hypercube of side length \(2^\ell\).
A hierarchical tree representation of \(\W\) is given by \(\T = (\mathcal E(\T),  \N(\T))\), which consists of a collection of nodes \(\N(\T)\) and a set of edges \(\mathcal E(\T)\), the latter of which specify the nodal interconnections~\cite{Bondy1976}.
The developments in this paper will be limited to hierarchical tree structures in the form of quadtrees, although the analysis holds for other  high-dimensional tree structures (e.g., octrees).
With this in mind, we let \(\T^\Q\) be the space of all feasible, multi-resolution, quadtree abstractions of \(\W\), and take \(\T_\W\in\T^\Q\) to be the finest-resolution tree representation of \(\W\); that is, the leafs of \(\T_\W\) coincide with the finest-resolution cells of \(\W\).
Moreover, for any \(\T\in\T^\Q\), we let \(\Nlf(\T)\) be the set of leaf nodes, \(\Nint(\T)\) be the collection of all interior (i.e., non-leaf) nodes, and, for any \(0\leq k\leq \ell\), \(\N_k(\T)\) be the collection of nodes in the tree \(\T\) at depth \(k\).
Moreover, for any node \(t\in\N(\T)\), take \(\chd(t)\) to be the collection of all immediate child nodes of \(t\), and let \(\T_{(t)}\) be the subtree of \(\T\) rooted at the node \(t\); in other words, \(\T_{(t)}\) is the portion of the tree \(\T\in\T^\Q\) that is descendant from the node \(t\).
For more information regarding trees and subtrees, we refer the reader to~\cite{larsson2020q,larsson2022generalized}. 
Bridging the gap between trees and signal encodes, we note that the source \(X\) in our problem are the finest-resolution grid cells of \(\W\).
Then, any tree \(\T_q\in\T^\Q\) may be represented by a corresponding encoder \(p_q(t|x)\), where \(p_q(t|x) = 1\) if and only if the finest-resolution cell \(x\in\Nlf(\T_\W)\) is aggregated to the node \(t\in\Nlf(\T_q)\)~\cite{larsson2022generalized}.
As a result, we may define the information contained in the tree according to the function
\begin{equation}\label{eq:treeXinfo}
	I_X(\T_q) = I(T_q;X) = \sum_{t,x} p_q(t,x) \log \frac{p_q(t,x)}{p_q(t) p(x)},
\end{equation}
where \(T_q : \Omega \to \Re\) is a compressed representation of \(X\) whose outcomes are defined by the leafs of the tree \(\T_q \in \T^\Q\) and has distribution \(p_q(t) = \sum_x p_q(t|x)p(x)\).
Similarly, we may quantify the amount of relevant information contained in the tree \(\T_q\in\T^\Q\) via \(I_Y(\T_q)\), which is defined according to
\begin{equation}\label{eq:treeYinfo}
    I_Y(\T_q) = I(T_q;Y) = \sum_{t,y} p_q(t,y) \log \frac{p_q(t,y)}{p_q(t) p(y)},
\end{equation}
where \(p_q(t,y) = \sum_x p_q(t|x) p(x,y)\).
Importantly, relations~\eqref{eq:treeXinfo} and~\eqref{eq:treeYinfo} can be written in terms of the interior node set of the tree \(\T_q\in\T^\Q\) as
\begin{equation}\label{eq:incrementalXinfoChange}
	I_X(\T) = \sum_{s \in \Nint(\T_q)} \Delta I_X(s),
\end{equation}
and
\begin{equation}\label{eq:incrementalYinfoChange}
	I_Y(\T) = \sum_{s \in \Nint(\T_q)} \Delta I_Y(s),
\end{equation}
where \(\Delta I_X(s)\) and \(\Delta I_Y(s)\) quantify the (non-negative) incremental change in information contributed by the interior node \(s\).
These incremental information contributions are, in turn, given by
\begin{equation}
	\Delta I_Y(s) = \JS(p(y|s'_1),\ldots,p(y|s'_n)),\quad s'_i\in\chd(s),
\end{equation}
and \(\Delta I_X(s) = H(\Pi)\), where \(\Pi = [p(s'_1)/p(s),\ldots,p(s'_{\lvert \chd(s)\rvert})/p(s)]^\tp\) and the distributions \(p(s)\) and \(p(y|s'_i)\) are computed from \(p(x,y)\)~\cite[p.~9]{larsson2022generalized}.
For a more detailed discussion regarding the connection between hierarchical trees, signal encoders, and information theory, the interested reader is referred to~\cite{larsson2022generalized,larsson2020q}.
%

\section{Problem Formulation}\label{sec:probStatment}

Throughout, we will assume that the environment \(\W\) and distribution \(p(x,y)\), which specifies the statistical relationship between the source \(X\) (i.e., the finest-resolution cells of \(\W\)) and the relevant information \(Y\), are both known.
We consider establishing a formal connection between two problems.
The first problem, which we will call the \emph{primal problem} is
%
\begin{equation}\label{eq:HCIBtreeProb}
	v(D) = \min \{I_X(\T) : I_Y(\T) \geq D,~ \T\in\T^\Q\},
\end{equation}
where \(D \geq 0\).
Before proceeding, we have the following definition which delineates the terminology we will employ with respect to problem~\eqref{eq:HCIBtreeProb}.
\begin{definition}
	Let \(D \geq 0\) be a given scalar. Then a tree \(\T\in\T^\Q\) is \emph{primal feasible}, or \emph{feasible} for short, if \(I_Y(\T) \geq D\).
	Moreover, the tree \(\T^*\in\T^\Q\) is \emph{primal optimal} if it is feasible and \(I_X(\T^*) = v(D)\).
\end{definition}

We will assume throughout that there exists a primal feasible solution, and that \(I(X;Y) > 0\) (i.e., that the environment contains some relevant information).
The second problem we consider is
%
\begin{equation}\label{eq:IBtreeProb}
f(\beta) = \min\{I_X(\T) - \beta I_Y(\T) : \T\in\T^\Q\},
\end{equation}
where  \(\beta \geq 0\) is a non-negative parameter that trades information retention and compression, akin to the mechanism of the IB problem~\eqref{eq:stdIBproblem}.
It is important to note that the problems~\eqref{eq:HCIBtreeProb} and~\eqref{eq:IBtreeProb} are not equivalent to the standard information-bottleneck (IB) problem~\eqref{eq:stdIBproblem}.
In contrast, the problems~\eqref{eq:HCIBtreeProb} and~\eqref{eq:IBtreeProb} leverage the IB principle to formulate a tree (graph) compression problem where the objective is to find a compressed representation of \(\W\) in the form of a multi-resolution tree (a discrete optimization problem).
As a result, the additional constraint \(\T\in\T^\Q\) in~\eqref{eq:HCIBtreeProb} and~\eqref{eq:IBtreeProb} renders the above problems unsolvable by methods developed for the IB problem or its variations.
Each of the above problems have appeared independently in the literature.
The hard-constrained problem~\eqref{eq:HCIBtreeProb} first appeared in~\cite{larsson2021information}, where it was shown that the problem may be realized as an (linear) integer program.
More recently, a detailed of discussion of~\eqref{eq:HCIBtreeProb} and its connections to integer programming, multi-objective optimization, Pareto optimality, and (convex) linear programming relaxation was presented in~\cite{larsson2022linearBC}.
In contrast, the problem~\eqref{eq:IBtreeProb} was introduced prior to the hard-constrained problem in~\cite{larsson2020q}, where the authors develop an algorithm, called Q-tree search, that is guaranteed to find an optimal tree solution to~\eqref{eq:IBtreeProb} for any value of \(\beta\geq 0\).
A generalization of the problem~\eqref{eq:IBtreeProb} has been considered in~\cite{larsson2022generalized}, where the authors develop a framework for multi-resolution tree abstraction with respect to multiple sources of information.
Interestingly, in~\cite{larsson2022generalized}, the authors show how their generalized framework may be used to form semantically-information driven multi-resolution tree abstractions for autonomous systems applications, and discuss how other IB-like formulations over the space of trees can be recovered as special cases of the generalized problem.

However, while the two problems~\eqref{eq:HCIBtreeProb} and~\eqref{eq:IBtreeProb} have both appeared in the literature, the relationship between the two is not fully understood.
Specifically, it has not been established if the two approaches to the information-theoretic tree abstraction problem are equivalent or not; that is, if for every value of \(D \geq 0\) there exists a \(\beta \geq 0\) so that the tree solution to the two problems are equivalent in terms of the trade-off of information-retention and compression, and vice versa.
It is therefore our objective in this paper to formally establish whether the formulations~\eqref{eq:HCIBtreeProb} and~\eqref{eq:IBtreeProb} are equivalent and what, if any, relationship exists between these problems.
The importance in establishing whether or not the two formulations are equivalent is two-fold.
Firstly, providing a rigorous understanding relationship between the two formulations~\eqref{eq:HCIBtreeProb} and~\eqref{eq:IBtreeProb} furnishes a deeper understanding of the benefits and drawbacks to each approach, and insight into when one formulation is preferable over the other.
Secondly, an analysis into the two formulations~\eqref{eq:HCIBtreeProb} and~\eqref{eq:IBtreeProb} may give rise to methods for the selection of \(\beta \geq 0\) in~\eqref{eq:IBtreeProb} as a function of the hard-constraint \(D\geq 0\) in~\eqref{eq:HCIBtreeProb}.
As we will see, the selection of the trade-off parameter \(\beta \geq 0\) so that the resulting solution to~\eqref{eq:IBtreeProb} is feasible in the primal problem is challenging.
Consequently, the development of a method to select \(\beta \geq 0\) as a function of \(D \geq 0\) would allow hierarchical multi-resolution tree abstractions that are feasible in~\eqref{eq:HCIBtreeProb} to be obtained by employing search algorithms like Q-tree search~\cite{larsson2020q} that do not require solving a large integer program.
%

\section{Relaxation and Duality Theory for Tree Abstractions} \label{sec:RelaxAndDualTrees}

In this section, we seek to establish a formal connection between problems~\eqref{eq:HCIBtreeProb} and~\eqref{eq:IBtreeProb}, in addition to investigating whether the two formulations are equivalent in our problem setting (i.e., hierarchical tree abstractions).
To do so, we will leverage concepts from Lagrangian relaxation and duality theory, and show how \(\beta \geq 0\) in the problem~\eqref{eq:IBtreeProb} may be viewed as a dual variable of the primal problem~\eqref{eq:HCIBtreeProb}.
To this end, we note that problem~\eqref{eq:IBtreeProb} assumes that a value of \(\beta \geq 0\) is given, and does not explicitly depend on \(D \geq 0\).
Instead, the problem~\eqref{eq:IBtreeProb} is generally solved for a variety of \(\beta\) to generate a family of abstractions that trade information-retention and compression, whereby larger values of \(\beta\) result with multi-resolution trees that are more informative regarding \(Y\) than those obtained at lower \(\beta\) values.
To establish an explicit connection between \(D\) and \(\beta\), we introduce the \emph{dual function}~\cite{geoffrion1971duality,Lemarechal2001,hiriart-urruty_convex_1993} of the hard-constrained, primal problem~\eqref{eq:HCIBtreeProb} with respect to the constraint \(I_Y(\T) \geq D\) given by
%
\begin{equation}\label{eq:dualFunctionTreeIB}
	d(\beta) = \min\{L(\T;\beta) : \T\in\T^\Q\},
\end{equation}
where 
\begin{equation}\label{eq:TreeLagrangian}
	L(\T;\beta) = I_X(\T) +\beta (D - I_Y(\T)).
\end{equation}
In~\eqref{eq:dualFunctionTreeIB} and~\eqref{eq:TreeLagrangian}, the non-negative variable \(\beta\) is called the \emph{Lagrange multiplier} (or dual variable) associated with the constraint \(I_Y(\T) \geq D\).
The importance of~\eqref{eq:dualFunctionTreeIB} is that, provided that the values of \(\beta\) and \(D\) are given and fixed, the solution to~\eqref{eq:IBtreeProb} is equivalent to the solution of~\eqref{eq:dualFunctionTreeIB}, since the additional term, \(\beta D\), does not depend on the tree \(\T\in\T^\Q\).
As a result, for given \(D\) and \(\beta\), we may view~\eqref{eq:IBtreeProb} as a problem akin to the problem~\eqref{eq:dualFunctionTreeIB}, where the optimal value of~\eqref{eq:IBtreeProb} is \(f(\beta) = d(\beta) - \beta D\), but whose (tree) solution is the same as~\eqref{eq:dualFunctionTreeIB} for all \(\beta\).
Therefore, for fixed \(\beta\) and \(D\), any algorithm that provides a solution to~\eqref{eq:IBtreeProb} may also be employed to solve~\eqref{eq:dualFunctionTreeIB}.
It should be noted that the dual function~\eqref{eq:dualFunctionTreeIB} may be viewed as a \emph{Lagrangian relaxation}~\cite{Geoffrion2010,Lemarechal2001} of the primal problem~\eqref{eq:HCIBtreeProb} with respect to the constraint \(I_Y(\T) \geq D\), since, for any \(D\geq 0\) we have \(d(\beta) \leq v(D)\) for all \(\beta \geq 0\).
From the viewpoint of Lagrangian relaxation theory, the multiplier \(\beta \geq 0\) represents a ``price-to-be-paid" for violations of the constraint \(I_Y(\T) \geq D\).
In this way, the problem~\eqref{eq:dualFunctionTreeIB} may be viewed as a relaxation of~\eqref{eq:HCIBtreeProb} where the constraint \(I_Y(\T) \geq D\) is moved into the objective with weighing parameter \(\beta \geq 0\), and the subsequent optimization performed over the entirety of \(\T^\Q\) as opposed to the trees in the set \(\{\T\in\T^\Q : I_Y(\T) \geq D\} \subseteq \T^\Q\).
Consequently, solutions to~\eqref{eq:dualFunctionTreeIB} (or~\eqref{eq:IBtreeProb}) for a given value of \(\beta \geq 0\) need not be feasible in~\eqref{eq:HCIBtreeProb}, and are therefore generally not primal optimal.
In light of the above observations, it is clear that the choice of \(\beta \geq 0\) is important in terms of the tree solutions obtained by solving~\eqref{eq:dualFunctionTreeIB}.
Motivated by the observation that \(d(\beta) \leq v(D)\) for all \(\beta \geq 0\) and \(D \geq 0\), the dual problem~\cite{boyd2004convex,hiriart-urruty_convex_1993}
aims to find a suitable setting of \(\beta\geq 0\) by considering the problem
\begin{equation}\label{eq:IBtreeDualProblem}
	\max\{d(\beta) : \beta  \geq 0\}.
\end{equation}
Notice that for any \(\beta^*\in\argmax\{d(\beta):\beta\geq 0\}\), we have \(d(\beta)\leq d(\beta^*) \leq v(D)\).
The utility of the dual function~\eqref{eq:dualFunctionTreeIB} and corresponding dual problem~\eqref{eq:IBtreeDualProblem} lies in the fact that the former is a convex function of the dual variable \(\beta \geq 0\) and that, as a result, the latter is a convex optimization problem~\cite{boyd2004convex}.
Combining these observations with the fact that the dual function~\eqref{eq:dualFunctionTreeIB} is a relaxation of the primal problem~\eqref{eq:HCIBtreeProb}, we note that the utility of the dual problem lies in replacing a difficult-to-solve, non-convex and non-linear optimization problem~\eqref{eq:HCIBtreeProb} with a convex one given by~\eqref{eq:IBtreeDualProblem}.
However, it should be noted that in order to employ~\eqref{eq:dualFunctionTreeIB} to solve the primal~\eqref{eq:HCIBtreeProb}, we require that there exists a value of \(\beta \geq 0\) so that the tree solution to~\eqref{eq:dualFunctionTreeIB} is optimal in~\eqref{eq:HCIBtreeProb}.
As we will see, such a value of the dual variable need not exist.
To this end, the inequality \(d(\beta)\leq v(D)\leq I_X(\T)\) for all \(\beta \geq 0\) and primal feasible trees \(\T\in\T^\Q\) is known as the \emph{weak duality theorem}~\cite[p.~149]{hiriart-urruty_convex_1993}.
Furthermore, if there exists \(\beta^* \geq 0\) such that \(d(\beta^*) = v(D)\) then strong duality holds.
In contrast to weak duality, strong duality does not hold in general.
However, the following theorem provides sufficient conditions for strong duality to be established.
\begin{theorem}[\hspace{-1sp}\cite{geoffrion1971duality}]\label{thm:sufficientForStrongDuality}
	Assume \(D \geq 0\) is a given scalar.
	If the pair \((\T^*, \beta^*)\), where \(\T^*\in\T^\Q\), satisfies:
	\begin{enumerate}
		\item[(i).] \(\T^* \in \argmin L(\T;\beta^*),\)
		\item[(ii).] \(I_Y(\T^*) \geq D ~\text{and}~ \beta^* \geq 0,\)
		\item[(iii).] \(\beta^* (D - I_Y(\T^*)) = 0,\)
	\end{enumerate}
	then \(\T^*\) is primal optimal, \(\beta^*\) is a solution to~\eqref{eq:IBtreeDualProblem}, and \(d(\beta^*) = v(D) = I_X(\T^*)\).
\end{theorem}
The conditions delineated by Theorem~\ref{thm:sufficientForStrongDuality} are called the \emph{optimality conditions}\cite{geoffrion1971duality}.
While Theorem~\ref{thm:sufficientForStrongDuality} provides sufficient conditions for strong duality to hold, we remark that the converse of Theorem~\ref{thm:sufficientForStrongDuality} is true if strong duality can be established, for example via Slater's constraint qualification in convex programming problems~\cite[p.~226]{boyd2004convex}.
That is, if strong duality holds, \(\T^*\in\T^\Q\) is primal optimal, and \(\beta^*\) is a solution to~\eqref{eq:IBtreeDualProblem} then conditions (i)-(iii) of Theorem~\ref{thm:sufficientForStrongDuality} hold~\cite{geoffrion1971duality}.

As we previously remarked, there need not, in general, exist a setting of the dual variable \(\beta \geq 0\) such that strong duality holds.
However, the important observation is that strong duality is a question of finding a value of \(\beta\geq 0\) for a given \(D \geq 0\) so that \(d(\beta) = v(D)\), since the converse is always true.
That is, for every \(\beta \geq 0\) there will always exist a setting of \(D \geq 0\) so that strong duality holds.
Before we state the theorem that formally establishes this observation, we remind the reader that, for given \(\beta \geq 0\), the tree solution to the problem~\eqref{eq:dualFunctionTreeIB} does not depend on \(D \geq 0\).
\begin{theorem}[\hspace{-1sp}{\cite[p.~163]{hiriart-urruty_convex_1993}}]\label{thm:betaToDthmEverett}
	Let \(\beta \geq 0\) be a given scalar and assume \(\T_\beta \in \T^\Q\) is the corresponding tree solution to~\eqref{eq:dualFunctionTreeIB}.
    Define \(D_\beta = I_Y(\T_\beta)\).
	Then \(\T_\beta\) is also an optimal solution to the problem
	\begin{equation*}
		v(D_\beta) = \min\{I_X(\T) : I_Y(\T) \geq D_\beta,~\T\in\T^\Q\},
	\end{equation*}
	and \(d(\beta) = v(D_\beta) = I_X(\T_\beta)\).
\end{theorem}
It should be noted that, with any tree \(\T_\beta \in \T^\Q\) obtained according to Theorem~\ref{thm:betaToDthmEverett}, the pair \((\T_\beta, \beta)\) satisfies (i)-(iii) of Theorem~\ref{thm:sufficientForStrongDuality}, since, by construction, \(\T_\beta \in \argmin\{L(\T;\beta):\T\in\T^\Q\}\), \(I_Y(\T_\beta) \geq D_\beta\), \(\beta \geq 0\), and \(\beta (D_\beta - I_Y(\T_\beta)) = 0\).

The importance of Theorem~\ref{thm:betaToDthmEverett} in our problem is that it establishes that the problems~\eqref{eq:HCIBtreeProb} and~\eqref{eq:dualFunctionTreeIB} have equivalent solutions (i.e., strong duality holds) for those values of \(D \geq 0\) which can be obtained by solving~\eqref{eq:dualFunctionTreeIB} (or, equivalently~\eqref{eq:IBtreeProb}) for a given value of \(\beta \geq 0\).
Our objective now is to determine whether strong duality holds in our problem for every value of \(D \geq 0\), not just those obtained in the manner delineated by Theorem~\ref{thm:betaToDthmEverett}.
To this end, we will next investigate the structure of solutions to~\eqref{eq:dualFunctionTreeIB} more closely to determine if the strong duality property holds in our problem.

Recall that, since the tree solution to problem~\eqref{eq:dualFunctionTreeIB} for a given value of \(\beta \geq 0\) does not depend on \(D \geq 0\), any method that finds a solution to~\eqref{eq:IBtreeProb} can be employed to solve~\eqref{eq:dualFunctionTreeIB}.
The above observation is more explicit if we write~\eqref{eq:dualFunctionTreeIB} as
\begin{equation}\label{eq:dualFunctionIBTrees}
	d(\beta) = \min\{I_X(\T) - \beta I_Y(\T): \T\in\T^\Q\} + \beta D.
\end{equation}
A solution to the problem \(\min\{I_X(\T) - \beta I_Y(\T): \T\in\T^\Q\}\) (i.e.,~\eqref{eq:IBtreeProb}) can be obtained by employing the Q-tree search algorithm~\cite{larsson2020q}.
To solve~\eqref{eq:IBtreeProb}, the Q-tree search algorithm prunes \(\T_\W\) by traversing the nodes of \(\T_\W\) and employing a node-wise pruning rule that relies on the so-called Q-function, defined according to the rule: if \(t\in\Nint(\T_\W)\) then
%
\begin{align}
Q(t;&\beta) = \nonumber\\
&~\min\{\Delta I_X(t) - \beta \Delta I_Y(t) + \sum_{t'\in\chd(t)}Q(t';\beta),~0\},\label{eq:QfunctionDef}
\end{align}
and where \(Q(t;\beta) = 0\) otherwise (i.e., \(t\in \Nlf(\T_\W)\)).
Observe that, by definition, \(Q(t;\beta)\leq 0\) for all \(t\in\N(\T_\W)\).
Q-tree search then builds an optimal solution to~\eqref{eq:IBtreeProb} by starting at the root node of the tree \(\T_\W\) and expanding nodes \(t\in\Nint(\T_\W)\) for which \(Q(t;\beta) < 0\) in a top-down manner, leaving those nodes \(t\in\N(\T_\W)\) with \(Q(t;\beta) = 0\) as leafs of the resulting tree solution.
The Q-tree search algorithm is guaranteed to find an optimal tree solution to~\eqref{eq:IBtreeProb}~\cite{larsson2020q}. 
Importantly, the value of the Q-function~\eqref{eq:QfunctionDef} may be related to the value of the objective in~\eqref{eq:IBtreeProb}.
To show how this is possible, we require the following.
\begin{lemma}[\hspace{-1sp}\cite{larsson2020q}]\label{lem:existenceOfTreeAndQfunction}
    Let \(\beta \geq 0\) be a given scalar, and \(t\) be any node in \(\N(\T_\W)\).
    Then \(Q(t;\beta) \leq \sum_{s \in \Nint(\T_{(t)})} \Delta I_X(s) - \beta \Delta I_Y(s)\) for all \(\T \in \T^\Q\).
    Moreover, there exits a tree \(\tilde \T\in\T^\Q\) such that \(Q(t;\beta) = \sum_{s\in\Nint(\tilde\T_{(t)})}\Delta I_X(s) - \beta \Delta I_Y(s)\).
\end{lemma}
See Lemma 1 of~\cite{larsson2020q} for the proof of Lemma~\ref{lem:existenceOfTreeAndQfunction}.
What this lemma says is that, for any node \(t\), there exists a tree \(\tilde\T \in \T^\Q\) for which node-wise the Q-value equals the sum of weighted incremental information contributions of all nodes descendant from \(t\) in the tree \(\tilde \T\).
In addition, Lemma~\ref{lem:existenceOfTreeAndQfunction} states that, for any other tree \(\T\in\T^\Q\), the Q-value for the node \(t\) is a lower bound to the weighted incremental information sum over all descendants of \(t\) in the (alternative) tree \(\T\).
We now have the following result, which relates the Q-function~\eqref{eq:QfunctionDef} to the value of~\eqref{eq:IBtreeProb} as a function of \(\beta\).
\begin{proposition}\label{prop:QfunctionAndObjIBTree}
	Let \(\tRoot \in \N(\T_\W)\) be the root node of the tree \(\T_\W\).
	Then
	\begin{equation*}
		Q(\tRoot;\beta) = \min\{I_X(\T) -\beta I_Y(\T):\T\in\T^\Q\}.
	\end{equation*}
\end{proposition}
\begin{proof}
    The proof is given in Appendix~\ref{app:QfunctionAndObjIBTreePrf}.
\end{proof}
By employing Proposition~\ref{prop:QfunctionAndObjIBTree}, we note that we can write the dual function~\eqref{eq:dualFunctionTreeIB} as 
\begin{equation}\label{eq:dualFunctionQ}
	d(\beta) = Q(\tRoot;\beta) + \beta D.
\end{equation}
It is important to note that since the Q-tree search method computes the Q-values according to~\eqref{eq:QfunctionDef} in a bottom-up manner directly from the input data \(p(x,y)\) and \(\W\), the value of~\eqref{eq:dualFunctionQ} can be readily determined for a given value of \(\beta\) by querying the root-node Q-value once the bottom-up pass is completed, and therefore does not require solving an optimization problem.
Moreover, a tree \(\T_\beta\in\T^\Q\) that attains the value of~\eqref{eq:dualFunctionQ} for a given value of \(\beta \geq 0\) may be obtained by executing the Q-tree search algorithm.
Furthermore, by applying Proposition~\ref{prop:QfunctionAndObjIBTree} and~\eqref{eq:incrementalXinfoChange}-\eqref{eq:incrementalYinfoChange} to~\eqref{eq:dualFunctionQ}, we see that the dual function~\eqref{eq:dualFunctionTreeIB} in our problem may be written as
\begin{equation}\label{eq:dualIBtreeFunctionWithDeltaInfo}
d(\beta) = \sum_{t\in\Nint(\T_\beta)} \Delta I_X(t) + \beta(D - \sum_{t\in\Nint(\T_\beta)} \Delta I_Y(t)),
\end{equation}
where \(\T_\beta \in\T^\Q\) is the tree solution to~\eqref{eq:dualFunctionTreeIB} for the given \(\beta \geq 0\).
Next, we introduce and discuss tree phase transitions, their relevance to our problem, and  develop an algorithm that leverages the mechanism of Q-tree search for their computation.
We will then show the connection of tree phase transitions to solutions of the dual problem~\eqref{eq:IBtreeDualProblem} and utilize our observations to establish that strong duality does not hold in our problem setting. 
%

\section{Phase Transitions, Solving the Dual Problem, and Connections to Greedy Search}\label{sec:PTsAndDualOpt}

In this section, we will utilize the ideas from the previous section in order to develop an algorithm that computes the \(\beta\geq 0\) value that solves the dual problem~\eqref{eq:IBtreeDualProblem}.
The algorithm we design exploits the structure of the dual function~\eqref{eq:dualFunctionTreeIB} via the relations~\eqref{eq:dualFunctionQ} and~\eqref{eq:dualIBtreeFunctionWithDeltaInfo}.
Before discussing the development of such an algorithm, we need to examine the Q-function~\eqref{eq:QfunctionDef} in detail to provide greater insight into our algorithmic developments that follow.

\subsection{Critical \(\beta\) values for Q-tree search}\label{subsec:criticalBeta}

By inspecting~\eqref{eq:dualFunctionQ}, it is clear that the Q-values play an important role in the structure of the dual function.
To this end, in this section, we discuss the Q-values in greater detail and define the concept of \(\beta\)-critical value that will play an important role in the development of an algorithm to solve the dual problem~\eqref{eq:IBtreeDualProblem}.
The following property of the Q-function is vital to our later developments.
\begin{proposition}\label{prop:QfunctionMonotone}
	The Q-function \(Q(t;\beta)\) is monotone decreasing with respect to \(\beta\).
	That is, if \(\beta_1 \geq \beta_2\) then \(Q(t;\beta_1)\leq Q(t;\beta_2)\) for all nodes \(t \in \N(\T_\W)\).
\end{proposition}
\begin{proof}
    The proof is given in Appendix~\ref{app:QfunctionMonotonePrf}.
\end{proof}
In light of Proposition~\ref{prop:QfunctionMonotone}, and the fact that \(Q(t;\beta) \leq 0\) for all \(\beta \geq 0\), we define the \emph{critical \(\beta\) value} for each node \(t\in\N(\T_\W)\) as
\begin{equation}\label{eq:betaCriticalQ}
	\betacrq(t) = 
	\begin{cases}
		\sup\{\beta \geq 0: Q(t;\beta) = 0\}, &\text{ if } t\in\Nint(\T_\W),\\
		\infty, & \text{ otherwise.}
	\end{cases}
\end{equation}
It should be noted that the critical \(\beta\) value defined in~\eqref{eq:betaCriticalQ} is distinct from the critical \(\beta\) value discussed in~\cite{larsson2022generalized}.
More specifically, the authors of~\cite{larsson2022generalized} define \(\beta\)-critical for any \(t\in\N(\T_\W)\) according to
\begin{equation}\label{eq:criticalBetaOneStep}
	\betacr(t) = 
	\begin{cases}
		\frac{\Delta I_X(t)}{\Delta I_Y(t)} & \text{ if } t\in\Nint(\T_\W) \text{ and } \Delta I_Y(t) > 0,\\
		\infty, & \text{ otherwise}. 
	\end{cases}
\end{equation}
The critical \(\beta\)-value defined in~\eqref{eq:criticalBetaOneStep} is the value of \(\beta \geq 0\) for which \(\Delta I_X(t) - \beta \Delta I_Y(t) = 0\).
However, the critical \(\beta\)-value defined according to the Q-function in~\eqref{eq:betaCriticalQ} captures the interplay of the one-step cost \(\Delta I_X(t) - \beta I_Y(t)\) \emph{and} the effect of the child Q-values \(Q(t';\beta)\) for \(t'\in\chd(t)\) on the value of \(\beta\)-critical. 

As a result of Proposition~\ref{prop:QfunctionMonotone}, it follows that for each node \(t\) such that \(\betacrq(t)<\infty\), the value of \(\betacrq(t)\) partitions the non-negative real line into two intervals; namely, the interval \([0,\betacrq(t)]\) for which \(Q(t;\beta) = 0\) and, the interval \((\betacrq(t),\infty)\) for which \(Q(t;\beta) < 0\).
If, instead, \(\betacrq(t) = \infty\) then all values of \(\beta \geq 0\) will result in \(Q(t;\beta) = 0\).
To establish a connection between \(\betacrq(t)\) and \(\betacr(t)\), we note that, since \(Q(t;\beta) \leq 0\) for all \(t\) and \(\beta \geq 0\), it follows that
\begin{align}\label{key}
	\Delta I_X(t) &- \beta \Delta I_Y(t) \nonumber\\
	&\geq \Delta I_X(t) - \beta \Delta I_Y(t) + \sum_{t'\in\chd(t)} Q(t';\beta), \nonumber\\
	&\geq Q(t;\beta),
\end{align}
where the final inequality follows from the definition of \(Q(t;\beta)\) in~\eqref{eq:QfunctionDef}.
Therefore, if \(t\) is a node such that \(\betacr(t) < \infty\), then we have \(\Delta I_X(t) - \betacr(t) \Delta I_Y(t) = 0\), and so \(Q(t;\betacr(t)) \leq 0\).
Therefore, by employing Proposition~\ref{prop:QfunctionMonotone}, we conclude that 
\(\betacrq(t) \leq \betacr(t)\).
Note that the above result can be straightforwardly extended to all nodes, as for any node \(t\) for which \(\betacr(t) = \infty\), it straightforwardly holds that \(\betacrq(t) \leq \betacr(t)\).

The observation that \(\betacrq(t) \leq \betacr(t)\) provides additional insight as to why Q-tree search is able to find solutions that other methods to solve~\eqref{eq:IBtreeProb}, such as a greedy one-step look-ahead approach, are not able to.
In more detail, a greedy approach towards solving~\eqref{eq:IBtreeProb} expands nodes not according to the Q-values, but rather with respect to the one-step cost \(\Delta I_X(t) - \beta \Delta I_Y(t)\).
Thus, a greedy approach to solving~\eqref{eq:IBtreeProb} will traverse nodes top-down, identically to Q-tree search but expand nodes \(t\) only if \(\Delta I_X(t) - \beta I_Y(t) < 0\).
While simple, a greedy approach is known to not find an optimal solution to~\eqref{eq:IBtreeProb}~\cite{larsson2022generalized}.
Both the Q-tree search and greedy approaches may be alternatively viewed in terms of the node-wise \(\betacrq(t)\) and \(\betacr(t)\) functions, respectively, instead of the values of \(Q(t;\beta)\) or the one-step cost \(\Delta I_X(t) - \beta\Delta I_Y(t)\).
Namely, for a given value of \(\beta \geq 0\), the Q-tree search and the greedy approach will expand nodes \(t\) that the algorithms reach for which \(\beta > \betacrq(t)\) and \(\beta > \betacr(t)\), respectively.
From this point of view, we see that, if for some node \(t\) we have \(\beta > \betacr(t)\), meaning that the node is expanded by the greedy approach, then it follows that \(\betacrq(t) \leq \betacr(t) < \beta\), and therefore the node is also expanded by Q-tree search.
However, if for some \(\beta \geq 0\) both algorithms reach a node \(t\) that satisfies \(\betacrq(t) < \beta \leq \betacr(t)\), then the node will be expanded by Q-tree search but not by the greedy search (see Section~6.1 of~\cite{larsson2022generalized}).
The above discussion provides insights as to why Q-tree search is able to find trees that are not obtainable by the one-step greedy approach and furnishes a dual perspective of existing results, such as Theorem~1 of~\cite{larsson2020q} (i.e., that the tree returned by the greedy search is a subtree of the tree found by Q-tree search).
Next, we will utilize these observations to develop an algorithm to solve the dual problem~\eqref{eq:IBtreeDualProblem}.
%

\subsection{Identifying Phase Transitions and an Algorithm to Solve the Dual Problem}\label{subsec:IDphaseTransitions}

To develop an algorithm to solve~\eqref{eq:IBtreeDualProblem} we will leverage our previously discussed observations, especially as they relate to so-called \emph{phase transitions} of the tree.
Specifically, a tree phase transition occurs when the tree solution to~\eqref{eq:IBtreeProb}, or equivalently~\eqref{eq:dualFunctionTreeIB}, changes as a function of \(\beta \geq 0\), as formalized by the following definition.
%
\begin{definition}\label{def:treePT}
    For any \(\beta \geq 0\), let \(\T_\beta \in\T^\Q\) be a minimal solution\footnote{A tree \(\tilde \T \in \T^\Q\) is a minimal solution to~\eqref{eq:IBtreeProb} if \(I_X(\T) - \beta I_Y(\T) < I_X(\tilde \T) -\beta I_Y(\tilde \T)\) for all trees \(\T\in\T^\Q\) that are subtrees of \(\tilde \T\). See~\cite{larsson2020q,larsson2022generalized} for more details.} to~\eqref{eq:IBtreeProb}.
    Then \(\beta\geq 0\) is a \emph{tree phase transition} if
    \begin{equation*}
        \Nint(\T_{\beta}) \subsetneq \Nint(\T_{\beta+\varepsilon}),
    \end{equation*}
    holds for all \(\varepsilon > 0\).
    The set comprising all tree phase transitions is denoted \(\betaPhaseTransitionSet\).
\end{definition}

We may draw a connection between tree phase transitions and the discussion in the previous section as follows. 
Namely, Definition~\ref{def:treePT} establishes that a tree phase transition is a value of \(\beta \geq 0\) for which the tree solution to problem~\eqref{eq:IBtreeProb} when the trade-off parameter has value \(\beta + \varepsilon\) is not the same as the tree solution to~\eqref{eq:IBtreeProb} of the unperturbed problem, where the \(\beta\)-perturbations may be arbitrarily small.
Since it is known that Q-tree search returns an optimal solution to~\eqref{eq:IBtreeProb} for any value of \(\beta \geq 0\), we may exploit the mechanism of Q-tree search to draw a connection between the Q-values of select nodes in \(\T_\beta\) and the values of tree phase transitions.

Specifically, we note that, if a tree solution, say \(\T_{\beta}\in\T^\Q\), is known for for some \(\beta \geq 0\) from executing Q-tree search, then we may gather the values of \(\betacrq(t)\) for each \(t\in\Nlf(\T_\beta)\).
It follows, by definition of \(\betacrq(t)\),  that the mechanism of the Q-tree search algorithm, and Proposition~\ref{prop:QfunctionMonotone} that the tree solution to~\eqref{eq:IBtreeProb} will not change for any \(\beta'\) in the interval \(\beta \leq \beta' \leq \bar\beta\), where \(\bar \beta = \min\{\betacrq(t):t\in\Nlf(\T_{\beta})\}\), since \(Q(t;\beta') = 0\) for all \(t\in\Nlf(\T_{\beta})\).
Consequently, the value of \(\bar \beta = \min\{\betacrq(t):t\in\Nlf(\T_{\beta})\}\) is a tree phase transition, since any perturbation, no matter how small, of \(\bar \beta\) will result in an expansion of a leaf node of \(\T_{\beta}\), thereby changing the tree.
Thus, to determine the tree phase transitions, it is sufficient to consider only those values of \(\beta\) in the set \(\betacrqset\). 
The challenge, however, is that the solutions \(\T_\beta\in\T^\Q\) to~\eqref{eq:IBtreeProb} for any \(\beta \geq 0\) are not known a priori.
If they were, the above procedure could be followed to generate \(\betaPhaseTransitionSet\).
Our challenge is therefore two-fold: (1) determine \(\betacrq(t)\) for any \(t\in\N(\T_\W)\), and (2) establish which values of \(\betacrq(t)\) constitute a tree phase transition.

Before proceeding, we recall that \(\betaPhaseTransitionSet\) will be a subset of \(\betacrqset\).
Therefore, not all critical-\(\beta\) values will constitute a phase transition.
The reason for this follows again from the mechanism of the Q-tree search method.
Namely, since it is possible for \(\betacrq(t) > \betacrq(t')\) for some ancestor \(t\) of \(t'\), and since Q-tree search expands nodes in a top-down fashion, in this case, that as soon as the value of \(\beta\) is large enough to cause expansion of the node \(t\), that is \(\beta > \betacrq(t)\), the descendant \(t'\) will be expanded should it be reached during the search.
As a result, the algorithm will not terminate with \(t'\) as a leaf node of any solution to~\eqref{eq:IBtreeProb} and thus the value of \(\betacrq(t')\) will not constitute a phase transition.
Consequently, \(\betaPhaseTransitionSet\) is a subset of \(\betacrqset\).
Our challenge now is to devise an algorithm that computes \(\betacrqset\) and determine which of those \(\betacrq(t)\in\betacrqset\) constitute phase transitions to build the set \(\betaPhaseTransitionSet\) from input data.
With the above discussion in mind, it is important to note that the phase transitions are vital to solving the dual problem~\eqref{eq:IBtreeDualProblem}.
To see why this is the case, we note that relation~\eqref{eq:dualIBtreeFunctionWithDeltaInfo} indicates that the dual function is a piece-wise linear function of \(\beta\), which is differentiable at all points except those for which a tree phase transition occurs.
Consequently, if we are able to identify for what values of \(\beta\) a tree phase transition occurs, then we may utilize information of the subgradient to find a value of \(\beta\) that solves~\eqref{eq:IBtreeDualProblem}.
Thus, the idea of our proposed algorithm is to identify the tree phase transitions in a recursive manner, thereby determining where the solution to~\eqref{eq:dualFunctionTreeIB} changes as a function of \(\beta\), and then employ this knowledge to find the setting of \(\beta\) that maximizes \(d(\beta)\) as a function of \(D\) by exploiting relation~\eqref{eq:dualIBtreeFunctionWithDeltaInfo}.

%
\begin{algorithm}
	\caption{\(\texttt{treePhaseTransitions}(p(x,y),\T_\W)\).}\label{alg:phaseTransitionAlgorithm}
	\KwData{\(p(x,y)\), \(\T_\W\)}
	\KwResult{\(\betaPhaseTransitionSetNode{\tRoot},~Q_{X,\tRoot},~Q_{Y,\tRoot}\)}
	
	Set \(Q_{X,t} = \varnothing, Q_{Y,t} = \varnothing, \betaPhaseTransitionSetNode{t} = \varnothing\) for all \(t\)\;

	Set \(k = \ell -1 \)\;
	
	\While{\(k \geq 0\)}{
		\For{\(t\in\N_k(\T_\W)\)}
		{
			compute \(\betacr(t)\) according to~\eqref{eq:criticalBetaOneStep}\;
			
			\((\bar{\mathcal B}_{\text{PT},t},\bar{Q}_{X,t},\bar{Q}_{Y,t}) = \texttt{subtreePTs}(t)\)\; \label{algPT:getsubtreeInfo}
			
			\If{\(\betacr(t) < \min\{\beta : \beta \in \bar{\mathcal B}_{\text{PT},t}\}\)}
			{
				\label{algPT:betacrCheck1}
				\(\betacrq(t) = \betacr(t)\)\;
				\(dX_{\text{PT}} = \Delta I_X(t), dY_{\text{PT}} = \Delta I_Y(t)\)\;
			}\Else{
			
			\(\bar{\mathcal B}_{\text{PT},t} \leftarrow \texttt{sort}(\bar{\mathcal B}_{\text{PT},t})\)\; \label{algPT:sortingBetaChildren}
			
			sort \(\bar{Q}_{X,t}\) and \(\bar{Q}_{Y,t} \) according to \(\bar{\mathcal B}_{\text{PT},t}\).
			
			\(dX \leftarrow \Delta I_X(t),~dY\leftarrow \Delta I_Y(t)\)\;

			\For{\(i = 1:(\lvert\bar{\mathcal B}_{\text{PT},t}\rvert)\)}{
				\((\beta_{\text{lb}},\beta_{\text{ub}}) \leftarrow \texttt{betaBnds}(\bar{\mathcal B}_{\text{PT},t},i)\)\;\label{algPT:betaBndsGen}
			
			\(dX = dX + [\bar{Q}_{X,t}]_i\)\;\label{algPT:runningXinformation}
			\(dY = dY + [\bar{Q}_{Y,t}]_i\)\;\label{algPT:runningYinformation}
			
			\If{\(dY > 0\)}{
				\(\bar{\beta} = dX / dY\)\;\label{algPT:generateCandidateBeta1}
			}\Else{
				\(\bar{\beta} = \infty\)\;\label{algPT:generateCandidateBeta2}
			}

			\If{\(\beta_{\text{lb}} \leq \bar{\beta} \leq \beta_{\text{ub}}\)}
			{
				\If{\(\beta_{\text{lb}} \leq \betacr(t) \leq \beta_{\text{ub}}\)}
				{
					\(\betacrq(t) = \min\{\bar{\beta},\betacr(t)\}\)\;\label{algPT:betacrCheck2}
					
					\If{\(\betacrq(t) < \betacr(t)\)}{
						\(dX_{\text{PT}} = dX\)\;
						\(~dY_{\text{PT}} = dY\)\;
					}\Else{
						\(dX_{\text{PT}} = \Delta I_X(t)\)\;
						\(dY_{\text{PY}} = \Delta I_Y(t)\)\;
					}
				}\Else{
					\(\betacrq(t) = \bar{\beta}\)\;
					\(dX_{\text{PT}} = dX\)\;
					\(dY_{\text{PT}} = dY\)\;
				}
				\texttt{\textbf{break}}\;\label{algPT:candidateSolnBreak1}
			}\ElseIf{\(\beta_{\text{lb}} \leq \betacr(t) \leq \beta_{\text{ub}}\)}
			{
				\(\betacrq(t) = \betacr(t)\)\;
				\(dX_{\text{PT}} = \Delta I_X(t), dY_{\text{PT}} = \Delta I_Y(t)\)\;
				\texttt{\textbf{break}}\;\label{algPT:candidateSolnBreak2}
			}			
			}	
		}		
			
			\(\mathcal I = \{j: \betacrq(t) \leq  [\bar{\mathcal B}_{\text{PT},t}]_j <\infty \}\)\;\label{algPT:determineChildPT}
			\((\betaPhaseTransitionSetNode{t}, Q_{X,t},Q_{Y,t})\leftarrow\texttt{setPTData}({\mathcal I}, \bar Q_{X,t}, \bar Q_{Y,t}, dX_{\text{PT}},dY_{\text{PT}})\)\;\label{algPT:setData}
			
		}
		\(k = k - 1\)\;
	}
\end{algorithm}
Our proposed algorithm to determine the tree phase transitions is illustrated in Algorithm~\ref{alg:phaseTransitionAlgorithm}.
The general idea of the algorithm is to traverse the tree \(\T_\W\) in a bottom-up manner to determine the value of \(\betacrq(t)\) for each node, and to subsequently decide which, if any, of the candidate phase transitions in the set \(\{\betacrq(t'):t'\text{ is a descendant of } t\}\) may be excluded from consideration.
To do so, Algorithm~\ref{alg:phaseTransitionAlgorithm} finds the \(X\)- and \(Y\)-information of the subtrees of the candidate solutions to~\eqref{eq:IBtreeProb} that is descendant from the node \(t\), and utilizes this information to determine a candidate phase transition. 
The aforementioned process can be done in a bottom-up manner, as described next.
%

%
\begin{figure}[t]
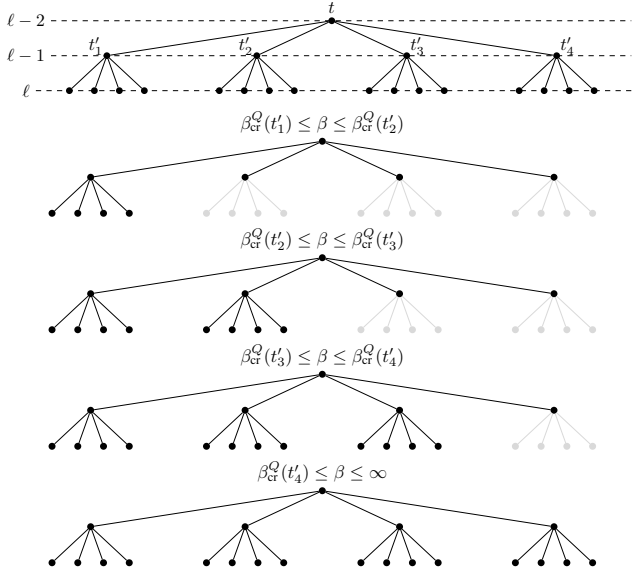

	\begin{adjustbox}{max size={0.48\textwidth}}

	\end{adjustbox}
	\caption{Example of algorithm progression when \(\betacrq(t'_1)\leq\betacrq(t'_2)\leq\betacrq(t'_3)\leq\betacrq(t'_4)\). Tree at top is the full tree rooted at \(t\) with node depth and labels.}
	\label{fig:ptAlgoSequenceSmall}
\end{figure}

\begin{figure*}[t]
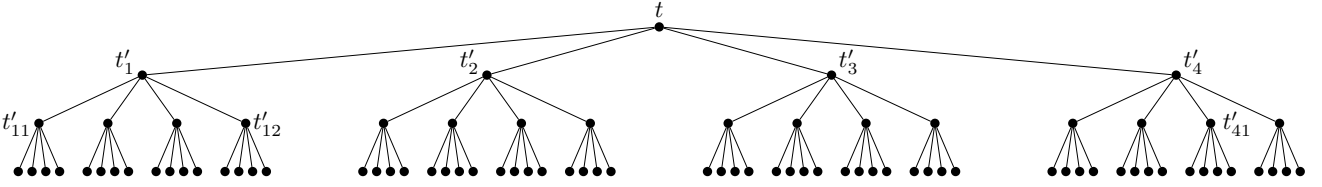

	\begin{adjustbox}{max size={0.98\textwidth}}
		%
	\end{adjustbox}
	\caption{Example of a subtree rooted at \(t\). While the value of \(\betacrq(t)\) is unknown, we assume that the values of \(\betacrq(t')\) for all descendants \(t'\) of \(t\) are known and satisfy \(\betacrq(t'_1)\leq\betacrq(t'_3)\leq\betacrq(t'_4)\leq\betacrq(t'_{11})\leq\betacrq(t'_{41})\leq\betacrq(t'_{12})\leq\betacrq(t'_2)\). Nodes not explicitly represented are assumed to have a critical \(\beta\)-value that is less than that of one of its ancestors, and therefore do not constitute a candidate phase transition.}
	\label{fig:ptAlgTreeMechanism}
\end{figure*}

%
To this end, Algorithm~\ref{alg:phaseTransitionAlgorithm} first considers nodes \(t\in\N_{\ell-1}(\T_\W)\), since \(\betacrq(t) = \infty\) for all \(t \in\Nlf(\T_\W)\) since the leafs of \(\T_\W\) are not expandable nodes. 
Thus, for any node \(t\in\N_{\ell-1}(\T_\W)\), we have that \(\betacrq(t) = \Delta I_X(t) / \Delta I_Y(t)\) if \(\Delta I_Y(t) > 0\), where \(\betacrq(t) = \infty\) otherwise, which follows from the definition of the Q-function.
That is, for \(t\in\N_{\ell-1}(\T_\W)\) we have, from~\eqref{eq:QfunctionDef},
\begin{equation}
    Q(t;\beta) = \min\{\Delta I_X(t) - \beta \Delta I_Y(t),~0\},
\end{equation}
and therefore \(\betacrq(t)\) is the value of the trade-off parameter that satisfies \(\Delta I_X(t) - \betacrq(t) \Delta I_Y(t) = 0\), which coincides with \(\betacr(t)\) for the node \(t\).
The associated \(X\)- and \(Y\)-information contributed by the phase transition (i.e., nodal expansion) is \(Q_{X,t} = \Delta I_X(t)\) and \(Q_{Y,t} = \Delta I_Y(t)\) if \(\betacrq(t) < \infty\), where \(Q_{X,t} = Q_{Y,t} = 0\) otherwise since, in the latter case, the node is never expanded.
Once the value of \(\betacrq(t)\) is known for all \(t\in\N_{\ell-1}(\T_\W)\), the algorithm moves to consider nodes at depth \(\N_{\ell-2}(\T_\W)\).
The process is more involved for the remaining nodes in the tree.
Namely, consider some \(t\in\N_{\ell-2}(\T_\W)\).
The node \(t\) has children \(\chd(t) = \{t'_1,\ldots,t'_4\}\), where each of the child nodes \(t'_i\) has a (known) value of \(\betacrq(t'_i)\), \(1\leq i\leq 4\), from the previous computations. 
In contrast to the the previous case when \(t\in\N_{\ell-1}(\T_\W)\), the case when \(t\in\N_{\ell-2}(\T_\W)\) must consider the changes in the Q-function at the node \(t\) brought about by the expansion of descendant nodes, the latter of which occur at the (known) values of \(\betacrq(t'_i)\).
To better understand the above discussion, let us assume, without loss of generality, that the values of \(\betacrq(t')\) satisfy \(\betacrq(t'_1)\leq \betacrq(t'_2) \leq \betacrq(t'_3)\leq\betacrq(t'_4)\).
Then, from the definition of \(\betacrq(t'_i)\), it follows that no child expansion occurs for \(\beta\in [0,\betacrq(t'_1)]\), at least one child expansion occurs for \(\beta\in(\betacrq(t'_1),\betacrq(t'_2)]\), and all children are expanded for \(\beta \in (\betacrq(t'_4),\infty)\).
What the above observations imply is that the values of the child phase transitions, namely \(\betacrq(t'_i)\), partition the non-negative real line into intervals that are characterized by the number of child expansions that occur for any \(\beta\) with a given interval.
For each of these intervals we then compute the \(X\)- and \(Y\)-information of the resulting subtree rooted at \(t\), and check if the candidate value of \(\betacrq(t)\) lies in the interval of \(\beta\) that is being considered.
More precisely, assume \(\beta\) is selected such that \(\betacrq(t'_1) < \beta \leq \betacrq(t'_2)\) and that, for the purposes of describing the algorithm, that the values of \(\betacrq(t'_i)\) are distinct.
Then, for the given value of \(\beta\) we have that the subtree rooted at \(t\) has a total \(X\)-information of \(\Delta I_X(t) + [\bar Q_{X,t}]_1\) and \(Y\)-information given by \(\Delta I_Y(t) + [\bar Q_{Y,t}]_1\), corresponding to the expansion of the node \(t\) and the child node \(t'_1\), which has the smallest value of \(\betacrq(t')\).
Thus, it follows that for \(\betacrq(t'_1) < \beta \leq \betacrq(t'_2)\),
\begin{align}
	&Q(t;\beta) = \nonumber \\
	&\min\{(\Delta I_X(t) + [\bar Q_{X,t}]_1) - \beta (\Delta I_Y(t) + [\bar Q_{Y,t}]_1),~0\} \label{eq:qFunctionAlgDescription1}.
\end{align}
We may then generate a candidate value of \(\betacrq(t)\), denoted \(\bar \beta\), by solving \((\Delta I_X(t) + [\bar Q_{X,t}]_1) - \bar{\beta} (\Delta I_Y(t) + [\bar Q_{Y,t}]_1) = 0\), which is done in Lines~\ref{algPT:generateCandidateBeta1} and~\ref{algPT:generateCandidateBeta2} of Algorithm~\ref{alg:phaseTransitionAlgorithm}.
If the resulting value of \(\bar \beta\) satisfies \(\betacrq(t'_1) < \bar \beta \leq \betacrq(t'_2)\), then \(\bar \beta\) is deemed a candidate value of \(\betacrq(t)\).
Otherwise, if \(\bar \beta\) does not fall within the interval \((\betacrq(t'_1),\betacrq(t'_2)]\) then it is not a candidate solution since the relation~\eqref{eq:qFunctionAlgDescription1} does not hold for \(\bar \beta \notin (\betacrq(t'_1),\betacrq(t'_2)]\).
In this case, we discard \(\bar \beta\) and check the next interval, \((\betacrq(t'_2),\betacrq(t'_3)]\), by repeating the aforementioned process, noting that there will now be two child expansions.
Since the interval \((\betacrq(t'_2),\betacrq(t'_3)]\) results in two child node expansions, the \(X\)-information and \(Y\)-information contributed by the subtree rooted at \(t\) is \(\Delta I_X(t) + [\bar Q_{X,t}]_1 + [\bar Q_{X,t}]_2\) and \(\Delta I_X(t) + [\bar Q_{Y,t}]_1 + [\bar Q_{Y,t}]_2\), respectively.
This computation is done in lines~\ref{algPT:runningXinformation} and~\ref{algPT:runningYinformation}, where we note that the variables \(dX\) and \(dY\) in Algorithm~\ref{alg:phaseTransitionAlgorithm} contain the cumulative sums of the \(X\)- and \(Y\)-information.
We then apply~\eqref{eq:qFunctionAlgDescription1} and check if a candidate solution is obtained, continuing this process until a solution is found. 
The above process is illustrated in~\figref{fig:ptAlgoSequenceSmall}.

In light of Proposition~\ref{prop:QfunctionMonotone}, and since we order the values of \(\betacrq(t')\) for descendants \(t'\) of \(t\) in increasing order in Line~\ref{algPT:sortingBetaChildren}, we terminate our search once a candidate solution is found in lines~\ref{algPT:candidateSolnBreak1} or~\ref{algPT:candidateSolnBreak2}.
However, the process we have thus far described assumes that a child node is indeed expanded.
Instead, it may be the case that no child Q-value contribution is required for \(Q(t;\beta)\) to undergo a transition from \(Q(t;\beta) = 0\) to \(Q(t;\beta) < 0\), which implies \(\betacrq(t) = \betacr(t)\).
Algorithm~\ref{alg:phaseTransitionAlgorithm} does not miss this case by comparing the candidate value of \(\betacrq(t)\) with that of \(\betacr(t)\) and ensuring that the minimum of these two values is ultimately declared the value of \(\betacrq(t)\).
The aforementioned checks are done in Lines~\ref{algPT:betacrCheck1} and~\ref{algPT:betacrCheck2} of Algorithm~\ref{alg:phaseTransitionAlgorithm}.
Once a solution is found, we store the corresponding values of the \(X\)- and \(Y\)-information contributed by the expansion of the node \(t\) when \(\beta \geq \betacrq(t)\) in the variables \(dX_{\text{PT}}\) and \(dY_{\text{PT}}\) so that they may be stored and used by ancestors of \(t\) at a later iteration of the algorithm.
See Figures~\ref{fig:ptAlgTreeMechanism} and~\ref{fig:ptAlgTreeMechanismSeq} for an illustration of the algorithm mechanism at a later stage in the recursive process.
%

\begin{figure}[tbh]
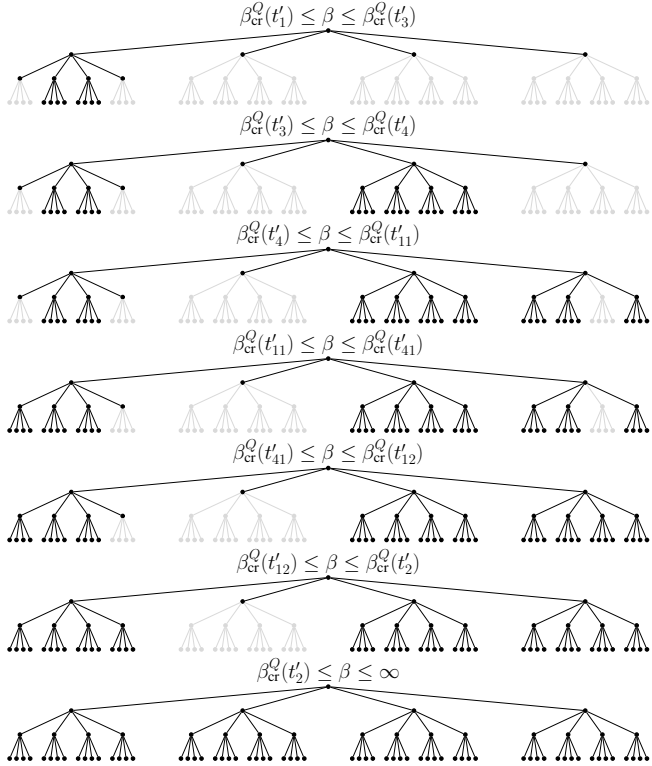

	\begin{adjustbox}{max size={0.48\textwidth}}

        \end{adjustbox}
	\caption{Possible subtrees rooted at \(t\) as \(\beta\) is increased. The subtrees are determined by the values \(\betacrq(t')\) of the descendants \(t'\) of \(t\) (see also~\figref{fig:ptAlgTreeMechanism}).}
	\label{fig:ptAlgTreeMechanismSeq}
\end{figure}

%
It is important to note, however, that only those \(\beta\)-values in the set \(\{\betacrq(s) \geq \betacrq(t): s\in\Nint(\T_{\W(t)})\}\) constitute a candidate tree phase transition.\footnote{We say \emph{candidate} tree phase transition here, since for \(\betacrq(s)\), \(s\in\N(\T_{\W(t)})\), to constitute a true tree phase transition, it must be the case that \(\betacrq(s) \geq \betacrq(\bar t)\) for \emph{all} ancestors \(\bar t\) of \(s\). 
The node \(t\) is only one of potentially many such nodes. 
Thus, while it may be the case that \(\betacrq(s) \geq \betacrq(t)\), it need not be the case that \(\betacrq(s) \geq \betacrq(\bar t)\) for every ancestor \(\bar t\) of \(s\). The tree phase transitions cannot be completely determined until \(\betacrq(\tRoot)\) has been computed.}
This is because if, for some \(s\in\N(\T_{\W(t)})\) we have \(\betacrq(s)\leq \betacrq(t)\), then the Q-tree search algorithm will expand \(s\) once it expands \(t\); should \(s\) be reached in the expansion process, and therefore \(s\) will not be a leaf node of any solution to~\eqref{eq:IBtreeProb}.
In light of these observations, we see that line~\ref{algPT:determineChildPT} determines which \(\betacrq(s)\) values constitute candidate phase transitions, and that the set \(\betaPhaseTransitionSetNode{t}\) contains the values of \(\betacrq(s)\) that satisfy \(\betacrq(s) \geq \betacrq(t)\) for \(s\in\N(\T_{\W(t)})\).
As a result, by repeating the above argument, we notice that \(\betaPhaseTransitionSetNode{\tRoot}\) contains the values of \(\betacrq(s)\) that satisfy \(\betacrq(s) \geq \betacrq(\tRoot)\) for \(s\in\N(\T_\W)\).
In other words, \(\betaPhaseTransitionSetNode{\tRoot}\) contains those values of \(\betacrq(t)\) that constitute tree phase transitions, as if \(\betacrq(t) \in \betaPhaseTransitionSetNode{\tRoot}\) then it must be the case that \(\betacrq(t) \geq \betacrq(s)\) for every ancestor \(s\) of \(t\).
Correspondingly, the associated quantities of \(Q_{X,\tRoot}\) and \(Q_{Y,\tRoot}\) characterize the total \(X\)- and \(Y\)-information of each of the tree phase transitions, and thus characterize the possible \(X\)- and \(Y\)-information combinations possible by any tree solution to~\eqref{eq:IBtreeProb}.
Notice also that \(\betacrq(\tRoot)\) is always a phase transition, since the root node, is by definition, the first node to be expanded in any solution to~\eqref{eq:IBtreeProb}.
With the above discussion in mind, it thus follows that \(\betaPhaseTransitionSet = \{[\betaPhaseTransitionSetNode{\tRoot}]_j:1\leq j \leq \texttt{length}(\betaPhaseTransitionSetNode{\tRoot})\}\).
As a result, we will, in the sequel, refer to \(\betaPhaseTransitionSetNode{\tRoot}\) and \(\betaPhaseTransitionSet\) interchangeably.
With the mechanism of Algorithm~\ref{alg:phaseTransitionAlgorithm} in mind, we provide a few brief comments regarding the routines invoked in Lines~\ref{algPT:getsubtreeInfo},~\ref{algPT:betaBndsGen}, and~\ref{algPT:setData} before moving to delineate how knowledge of the tree phase transitions may be utilized to obtain a solution to the dual problem~\eqref{eq:IBtreeDualProblem}.
Namely, calling \(\texttt{subtreePTs}(t)\) aggregates the child node information so that the set \(\bar{\mathcal B}_{\text{PT},t}\) contains the intervals that are to be considered in generating a candidate solution for \(\betacrq(t)\), and the vectors \(\bar Q_{X,t}\), \(\bar Q_{Y,t}\) contain the corresponding \(X\)- and \(Y\)-information contributions each of the corresponding subtrees rooted at the children of \(t\).
It is important to note that \(\bar Q_{X,t}\) and \(\bar Q_{Y,t}\) are not the same as \(Q_{X,t}\) and \(Q_{Y,t}\), as the latter contain only the \(X\)- and \(Y\)-information contributions of those subtrees that are brought about by values of \(\beta\) that are greater than the computed value of \(\betacrq(t)\).
In other words, \(\bar Q_{X,t}\) and \(\bar Q_{Y,t}\) may be viewed as candidate phase transition subtrees, whereas those values in \(Q_{X,t}\) and \(Q_{Y,t}\) correspond to subtrees that are to be considered by ancestor nodes when computing their critical \(\beta\)-values at a later iteration.
Similarly, \(\bar{\mathcal B}_{\text{PT},t}\) contains those intervals which must be checked when computing \(\betacrq(t)\) for the node \(t\), although there may be values \(\beta\in\betaPhaseTransitionSetNode{t}\) that satisfy \(\beta \leq \betacrq(t)\) once the value of \(\betacrq(t)\) is known.
Those \(\beta\in\betaPhaseTransitionSetNode{t}\) for which \(\beta \leq \betacrq(t)\) are not included in \(\betaPhaseTransitionSetNode{t}\), the latter of which is passed to the ancestors of \(t\).
To this end, once the value of \(\betacrq(t)\) is known, the routine \(\texttt{setPTData}(\cdot)\) in Line~\ref{algPT:setData} of Algorithm~\ref{alg:phaseTransitionAlgorithm} picks off those values of \(\beta\) satisfying \(\beta > \bar{\mathcal B}_{\text{PT},t}\) from \(\bar{\mathcal B}_{\text{PT},t}\) (summarized by the index set \(\mathcal I\)) and the corresponding values of \(\bar{Q}_{X,t}\) and \(\bar Q_{Y,t}\) that correspond to candidate phase transitions in order to pass them to the next iteration of the algorithm.
Lastly, the routine \(\texttt{betaBnds}(\cdot)\) in Line~\ref{algPT:betaBndsGen} sets the endpoints of the current \(\beta\)-interval under consideration.
Next, we discuss how the tree phase transitions may be utilized to generate a solution to the dual problem.
%

\subsection{Tree Phase Transitions and the Dual Problem}

It is important to notice that the discussion in the previous section regarding tree phase transitions and the computation thereof was completely independent of \(D\).
This was possible due to the relation~\eqref{eq:dualFunctionQ}, which shows that the dual function may be written as a sum of two terms: \(Q(\tRoot,\beta)\), which does not depend on \(D\), and \(\beta D\).
The tree phase transitions are then the values of \(\beta\) for which \(Q(\tRoot;\beta)\) changes value, which follows directly from Definition~\ref{def:treePT}, Proposition~\ref{prop:QfunctionAndObjIBTree} and the definition of critical \(\beta\)-value in~\eqref{eq:betaCriticalQ}.
Therefore, the tree phase transitions provide vital information in obtaining a solution to the dual problem~\eqref{eq:IBtreeDualProblem}.
%

%
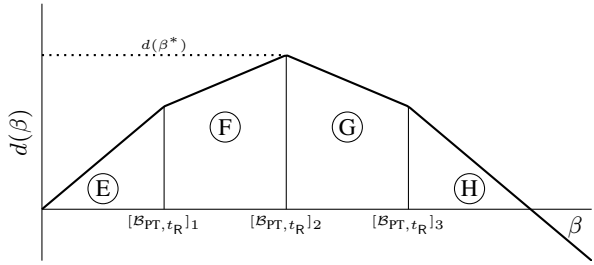
\begin{figure}[t]
	\centering
	\begin{tikzpicture}[
		declare function={
			func(\x)= (\x < 1) * (\x)   +
			and(\x >= 1, \x <= 2) * (1/2*\x + 1/2)     +
			and(\x > 2, \x <= 3) * (-1/2*\x + 5/2) + 
			(\x > 3) * (-1*\x + 4)
			;
		}
		]
		\begin{axis}[
			width=0.5\textwidth,
			height=5cm,
			axis x line*=middle, 
			axis y line*=middle,
			ymin=-0.5, 
			ymax=2, 
			ylabel={\small $d(\beta)$},
			xmin=0, 
			xmax=4.5, 
			xlabel={\small $\beta$},
			domain=0:5,
			samples=200, 
			ticks=none,
			x label style={at={(axis description cs: 0.97,0.21)},anchor=north},
			]
			
			\addplot [black,thick] {func(x)};
			
			\addplot [black, mark=none, opacity=0.4] coordinates {(1, 0) (1, 1)};
			\addplot [black, mark=none, opacity=0.4] coordinates {(2, 0) (2, 3/2)};
			\addplot [black, mark=none, opacity=0.4] coordinates {(3, 0) (3, 1)};
			
			\addplot [black, mark=none, dotted, thick] coordinates {(0, 3/2) (2, 3/2)};
			
			\node[circle, draw, minimum size=0.4cm, inner sep = 0pt] at (0.5,0.2) {\small E};
			\node[circle, draw, minimum size=0.4cm, inner sep = 0pt] at (1.5, 0.8) {\small F};
			\node[circle, draw, minimum size=0.4cm, inner sep = 0pt] at (2.5, 0.8) {\small G};
			\node[circle, draw, minimum size=0.4cm, inner sep = 0pt] at (3.5,0.2) {\small H};
			
			\node at (1,-0.15) {\tiny \([\betaPhaseTransitionSetNode{\tRoot}]_1\)};
			\node at (2,-0.15) {\tiny \([\betaPhaseTransitionSetNode{\tRoot}]_2\)};
			\node at (3,-0.15) {\tiny \([\betaPhaseTransitionSetNode{\tRoot}]_3\)};
			
			\node at (1,1.6) {\tiny $d(\beta^*)$};
		\end{axis}
	\end{tikzpicture} 
	\caption{Dual function structure for an environment \(p(x,y)\) and \(\W\) with three phase transitions, indicated by the corresponding elements of \(\betaPhaseTransitionSetNode{\tRoot}\) for some value of \(D \geq 0\).}
	\label{fig:dualFunctionVisualization}
\end{figure}

To this end, the set \(\betaPhaseTransitionSet\) can be readily computed from input data according to Algorithm~\ref{alg:phaseTransitionAlgorithm} discussed in the previous section, and will be henceforth represented by \(\betaPhaseTransitionSetNode{\tRoot}\).
Moreover, from executing the tree phase transition algorithm, we obtain the quantities of \(Q_{X,\tRoot}\) and \(Q_{Y,\tRoot}\), which quantify the \(X\)- and \(Y\)-information retained by each of the minimal tree solution to~\eqref{eq:IBtreeProb} associated with the phase transition values in \(\betaPhaseTransitionSetNode{\tRoot}\).

Specifically, once the tree phase transitions are known we may exploit the structural knowledge of the dual function~\eqref{eq:dualFunctionQ} to solve~\eqref{eq:IBtreeDualProblem} as a function of \(D\).
Namely, from~\eqref{eq:dualIBtreeFunctionWithDeltaInfo} we see that, for given \(D \geq 0\), the dual function~\eqref{eq:dualFunctionTreeIB} is piece-wise linear and differentiable, except for those values of \(\beta \geq 0\) where the tree solution to~\eqref{eq:IBtreeProb} changes; that is, \(d(\beta)\) is differentiable at all points except for \(\beta\in\betaPhaseTransitionSet\).

However, the values of \(\beta\) where the tree solution to~\eqref{eq:IBtreeProb} changes are, by definition, those values of \(\beta\) that phase transitions occur, which are contained in the set \(\betaPhaseTransitionSetNode{\tRoot}\) and can be obtained by Algorithm~\ref{alg:phaseTransitionAlgorithm}.
Moreover, Algorithm~\ref{alg:phaseTransitionAlgorithm} not only furnishes the set of tree phase transitions, but the total \(X\)- and \(Y\)-information of the trees associated with each of the phase transitions in \(\betaPhaseTransitionSetNode{\tRoot}\).
Consequently, from Definition~\ref{def:treePT} and relation~\eqref{eq:dualIBtreeFunctionWithDeltaInfo} we have that, for any \(\beta \geq 0\) the dual function~\eqref{eq:dualFunctionTreeIB} may be written as
\begin{equation}\label{eq:dualFunctionInTermsOfPTalgVars}
	d(\beta) = [Q_{X,\tRoot}]_{j-1} + \beta (D - [Q_{Y,\tRoot}]_{j-1}),
\end{equation}
where 
\begin{equation}
j = 
\begin{cases}
    \min\{i : \beta \leq [\betaPhaseTransitionSetNode{\tRoot}]_i\}, & \text{ if } \beta \leq \max(\betaPhaseTransitionSetNode{\tRoot}), \\
    \lvert \betaPhaseTransitionSetNode{\tRoot} \rvert + 1, & \text{ otherwise},
\end{cases}
\end{equation}
where \(\max(\betaPhaseTransitionSetNode{\tRoot})\) is the maximal phase transition computed by Algorithm~\ref{alg:phaseTransitionAlgorithm}, and \([Q_{X,\tRoot}]_0 := 0\), and \([Q_{Y,\tRoot}]_0 := 0\); see~\figref{fig:dualFunctionVisualization}.

From~\figref{fig:dualFunctionVisualization} we see that the tree phase transitions partition the \(\beta\)-axis into intervals where the slope of \(d(\beta)\) changes according to~\eqref{eq:dualFunctionInTermsOfPTalgVars}.
More specifically, in region \(\mathrm{E}\), that is, for \(0 \leq \beta \leq [\betaPhaseTransitionSetNode{\tRoot}]_1\), no expansion occurs and \(d(\beta) = \beta D\).
The corresponding \(X\)- and \(Y\)-information of the tree obtained for any \(\beta\) in this interval will be zero (recall that the solution to~\eqref{eq:IBtreeProb} in this interval is the root tree).
For the interval \([\betaPhaseTransitionSetNode{\tRoot}]_1 \leq \beta \leq [\betaPhaseTransitionSetNode{\tRoot}]_2\) (i.e., region F), the dual function is given by~\eqref{eq:dualFunctionInTermsOfPTalgVars} as \(d(\beta) = [Q_{X,\tRoot}]_1 + \beta (D - [Q_{Y,\tRoot}]_1)\).
Thus, the dual function is a linear function with slope \((D - [Q_{Y,\tRoot}]_1)\) and y-intercept \([Q_{X,\tRoot}]_1\).
The above argument may be repeated for the other intervals.

It is important to notice that, in~\figref{fig:dualFunctionVisualization}, we have drawn the dual function with \(\betacrq(\tRoot) = [\betaPhaseTransitionSetNode{\tRoot}]_1 > 0\).
In fact, \(\betacrq(\tRoot) > 0\) is always true, since from the data processing inequality\footnote{recall: \(T \leftrightarrow X \leftrightarrow Y\).}, we have \(I(T;X) \geq I(T;Y)\) (or, equivalently \(I_X(\T) \geq I_Y(\T)\)). 
Thus, \((1-\beta)I_X(\T) \leq I_X(\T) - \beta I_Y(\T) \leq I_X(\T)\) for all \(\T\in\T^\Q\). 
Consequently, for \(0 \leq \beta \leq 1\), \(I_X(\T) - \beta I_Y(\T)\) is minimized by the root tree for which \(I_X(\tRoot) - \beta I_Y(\tRoot) = 0\).
But since the root tree does not expand the root node, it follows that \(\betacrq(\tRoot) \geq 1 > 0\).

Leveraging the relation~\eqref{eq:dualFunctionInTermsOfPTalgVars}, our proposed algorithm to find a solution to the dual problem~\eqref{eq:IBtreeDualProblem} is shown in Algorithm~\ref{alg:subgradOptPhaseTrans}.
The proposed approach circumnavigates the non-differentiability of the dual function by employing sub-gradients.
Specifically, one may observe that the dual function may be written as
\begin{align}
	&d(\beta) = \nonumber \\
	&\max\{w : w \leq I_X(\T_{j}) +\beta(D - I_Y(\T_{j})),~ 0\leq j \leq \lvert \betaPhaseTransitionSet \rvert\},\label{eq:dualFunctionPTsMinPointwise}
\end{align}
where \(I_X(\T_0) = I_Y(\T_{0}) = 0\), and for \(1 \leq j \leq \lvert \betaPhaseTransitionSet\rvert\), \(\T_j\in\T^\Q\) is the tree satisfying \(I_X(\T_j) = [Q_{X,\tRoot}]_j\) and \(I_Y(\T_j) = [Q_{Y,\tRoot}]_j\).
As a result, at a phase transition, \(d(\beta) = I_X(\T_{j-1}) +\beta(D - I_Y(\T_{j-1}))\), and \(d(\beta) = I_X(\T_{j}) + \beta(D - I_Y(\T_{j}))\) for some \(1\leq j \leq \lvert \betaPhaseTransitionSetNode{\tRoot}\rvert\).
By employing~\eqref{eq:dualFunctionPTsMinPointwise}, it is straightforward to show that, for any such point, both \((D - I_Y(\T_{j-1}))\) \emph{and} \((D - I_Y(\T_{j}))\) are subgradients of the dual function at \([\betaPhaseTransitionSetNode{\tRoot}]_{j}\).

%
\begin{algorithm}[t]
	\caption{\(\texttt{optIBTreeDualFunct}(p(x,y),\T_\W,D)\), which finds a solution to the dual problem.}\label{alg:subgradOptPhaseTrans}
	\KwData{\(p(x,y)\), \(\T_\W\) constructed from \(\W\), \(D \geq 0\)}
	\KwResult{A solution \(\beta^* \geq 0\) to~\eqref{eq:IBtreeDualProblem}, optimal value \(d(\beta^*)\)}

    \((\betaPhaseTransitionSetNode{\tRoot},Q_{X,\tRoot},Q_{Y,\tRoot}) \leftarrow \texttt{treePhaseTransitions}(p(x,y),\T_\W)\)\; \label{algOPTdualFunc:getPTs}
    \(\betaPhaseTransitionSetNode{\tRoot} \leftarrow \texttt{sort}(\betaPhaseTransitionSetNode{\tRoot})\)\;\label{algOPTdualFunc:orderPTs}
    order \(Q_{Y,\tRoot}\) and \(Q_{X,\tRoot}\) according to \(\betaPhaseTransitionSetNode{\tRoot}\)\;\label{algOPTdualFunc:orderQs}
   	\If{\(D < [Q_{Y,\tRoot}]_1\)}
   	{
   		\label{algOPTdualFunc:subgradCheck1}
   		\(j^* = 0\)\;
   	}\ElseIf{\(D  \geq [Q_{Y,\tRoot}]_{j}\) for all \(j\)}
   	{
   		\label{algOPTdualFunc:subgradCheck2}
   		\(j^* = \texttt{length}(Q_{Y,\tRoot}) - 1\)\;
   	}\Else{
   		\label{algOPTdualFunc:subgradCheck3}
   		\(j^* = \max\{j: D - [Q_{Y,\tRoot}]_{j} \geq 0\}\)\;
   	}
    \(\beta^* = [\betaPhaseTransitionSetNode{\tRoot}]_{j^*+1}\)\;
    
    \If{\(j^* = 0\)}
    {
    	\(d(\beta^*) = \beta^* D\)\;
    }\Else{
		\(d(\beta^*) = [Q_{X,\tRoot}]_{j^*} + \beta^*(D - [Q_{Y,\tRoot}]_{j^*})\)\;
	}
    
\end{algorithm}
The mechanism of Algorithm~\ref{alg:subgradOptPhaseTrans} is as follows.
First, given an environment, that is, \(p(x,y),~\W\) and finest-resolution tree \(\T_\W\), Algorithm~\ref{alg:subgradOptPhaseTrans} retrieves the tree phase transitions by calling the phase transition algorithm in line~\ref{algOPTdualFunc:getPTs}.
Along with the phase transitions (contained in \(\betaPhaseTransitionSetNode{\tRoot}\)), we also obtain the corresponding \(X\)- and \(Y\)-information of each of the trees defined by the phase transitions.
We then sort the phase transitions in increasing order in line~\ref{algOPTdualFunc:orderPTs}, and ensure that the ordering of \(Q_{X,\tRoot}\) and \(Q_{Y,\tRoot}\) are consistent in line~\ref{algOPTdualFunc:orderQs}.
The algorithm then determines if the dual function will increase or decrease at a phase transition based on subgradient information.
Namely, if the algorithm determines that the dual function will decrease as \(\beta\) is increased beyond some phase transition, then the phase transition is declared a local minimum (which is also a global solution since the function \(d(\beta)\) is concave~\cite{Papadimitriou1998combinatorial}).
Finding such a local solution is done by the conditional statements in lines~\ref{algOPTdualFunc:subgradCheck1},~\ref{algOPTdualFunc:subgradCheck2}, and~\ref{algOPTdualFunc:subgradCheck3}.

Algorithm~\ref{alg:subgradOptPhaseTrans} employs the structural knowledge of the dual function.
That is, since the first tree phase transition, namely \(\betacrq(\tRoot) = \min(\betaPhaseTransitionSetNode{\tRoot})\), satisfies \(\betacrq(\tRoot) > 0\), the dual function will be monotone increasing on \([0,\betacrq(\tRoot)]\).
Therefore, if \(D - [Q_{Y,\tRoot}]_1 < 0\) then zero will be a subgradient of the dual function at \(\betacrq(\tRoot)\), as shown in~\figref{fig:dualFuncOPTcase1}.
Consequently, the dual function is maximized at \(\betacrq(\tRoot)\) and has optimal value \(\betacrq(\tRoot)D\), which is considered in line~\ref{algOPTdualFunc:subgradCheck1} (recall that \(\betacrq(\tRoot) = [\betaPhaseTransitionSetNode{\tRoot}]_1\)).
In the other extreme case, namely \(D = I_Y(\T_\W)\), we have that \(D - [Q_{Y,\tRoot}]_j \geq 0\) for all \(j\), since any tree, other than one that retains all the relevant information in the environment, will have less than \(D\) units of information, and results with a dual function structure as shown in~\figref{fig:dualFuncOPTcase2}.
Algorithm~\ref{alg:subgradOptPhaseTrans} checks for this case in line~\ref{algOPTdualFunc:subgradCheck2}, and returns the maximal element of \(\betaPhaseTransitionSetNode{\tRoot}\) as a solution if the conditional statement in line~\ref{algOPTdualFunc:subgradCheck2} is true.
Finally, for all other cases, our algorithm checks in line~\ref{algOPTdualFunc:subgradCheck3} for a phase transition where zero is a subgradient of the dual function.
Once such a phase transition is found, our algorithm returns the corresponding value of \(\beta\) and the dual function value.

%
%
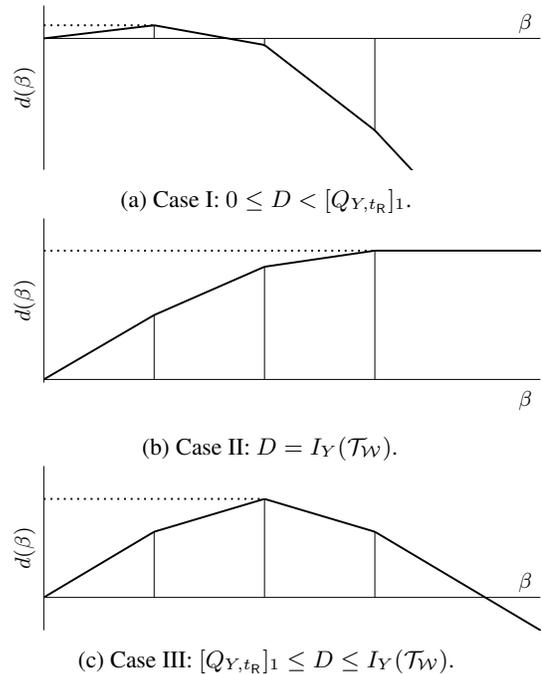
\begin{figure}[tbh]
	\centering
	\subfloat[Case I: {\(0\leq D < [Q_{Y,\tRoot}]_1\).} \label{fig:dualFuncOPTcase1}]
	{
		\begin{adjustbox}{max size={0.4\textwidth}}
			\begin{tikzpicture}[
				declare function={
					func(\x)= (\x < 1) * ((0.2)*\x)   +
					and(\x >= 1, \x <= 2) * ((0.2 - 1/2)*\x + 1/2)     +
					and(\x > 2, \x <= 3) * ((0.2 - 3/2)*\x + 5/2) + 
					(\x > 3) * ((0.2 - 2)*\x + 4)
					;
				}
				]
				\begin{axis}[
					width=0.5\textwidth,
					height=4cm,
					axis x line*=middle, 
					axis y line*=middle,
					ymin=-2, 
					ymax=0.5, 
					xmin=0, 
					xmax=4.5, 
					domain=0:5,
					samples=200, 
					ticks=none,
					ylabel={\small $d(\beta)$},
					xlabel={\small $\beta$},
					x label style={at={(axis description cs: 0.97,1.0)},anchor=north},
					]
					
					\addplot [black,thick] {func(x)};
					
					\addplot [black, mark=none, opacity=0.4] coordinates {(1, 0) (1, 0.2)};
					\addplot [black, mark=none, opacity=0.4] coordinates {(2, 0) (2, -0.1)};
					\addplot [black, mark=none, opacity=0.4] coordinates {(3, 0) (3, -1.4)};
            
					\addplot [black, mark=none, dotted, thick] coordinates {(0, 0.2) (1, 0.2)};
				\end{axis}
			\end{tikzpicture} 
		\end{adjustbox}
	}
	\\
	\subfloat[Case II: {\(D = I_Y(\T_\W)\).}\label{fig:dualFuncOPTcase2}]
	{
		\begin{adjustbox}{max size={0.4\textwidth}}
			\begin{tikzpicture}[
				declare function={
					func(\x)= (\x < 1) * ((2)*\x)   +
					and(\x >= 1, \x <= 2) * ((2 - 1/2)*\x + 1/2)     +
					and(\x > 2, \x <= 3) * ((2 - 3/2)*\x + 5/2) + 
					(\x > 3) * ((2 - 2)*\x + 4)
					;
				}
				]
				\begin{axis}[
					width=0.5\textwidth,
					height=4cm,
					axis x line*=middle, 
					axis y line*=middle,
					ymin=-0.1, 
					ymax=5, 
					xmin=0, 
					xmax=4.5, 
					domain=0:5,
					samples=200, 
					ticks=none,
					ylabel={\small $d(\beta)$},
					xlabel={\small $\beta$},
					x label style={at={(axis description cs: 0.97,0.0)},anchor=north},
					]
					
					\addplot [black,thick] {func(x)};
					
					\addplot [black, mark=none, opacity=0.4] coordinates {(1, 0) (1, 2)};
					\addplot [black, mark=none, opacity=0.4] coordinates {(2, 0) (2, 3.5)};
					\addplot [black, mark=none, opacity=0.4] coordinates {(3, 0) (3, 4)};
					
					\addplot [black, mark=none, dotted, thick] coordinates {(0, 4) (3, 4)};
				\end{axis}
			\end{tikzpicture} 
		\end{adjustbox}
	}
	\\
	\subfloat[Case III: {\([Q_{Y,\tRoot}]_1 \leq D \leq I_Y(\T_\W)\).} \label{fig:dualFuncOPTcase3}]
	{
		\begin{adjustbox}{max size={0.4\textwidth}}
			\begin{tikzpicture}[
				declare function={
					func(\x)= (\x < 1) * (\x)   +
					and(\x >= 1, \x <= 2) * (1/2*\x + 1/2)     +
					and(\x > 2, \x <= 3) * (-1/2*\x + 5/2) + 
					(\x > 3) * (-1*\x + 4)
					;
				}
				]
				\begin{axis}[
					width=0.5\textwidth,
					height=4cm,
					axis x line*=middle, 
					axis y line*=middle,
					ymin=-0.5, 
					ymax=2, 
					xmin=0, 
					xmax=4.5, 
					domain=0:5,
					samples=200, 
					ticks=none,
					ylabel={\small $d(\beta)$},
					xlabel={\small $\beta$},
					x label style={at={(axis description cs: 0.97,0.4)},anchor=north},
					]
					
					\addplot [black,thick] {func(x)};
					
					\addplot [black, mark=none, opacity=0.4] coordinates {(1, 0) (1, 1)};
					\addplot [black, mark=none, opacity=0.4] coordinates {(2, 0) (2, 3/2)};
					\addplot [black, mark=none, opacity=0.4] coordinates {(3, 0) (3, 1)};
					
					\addplot [black, mark=none, dotted, thick] coordinates {(0, 3/2) (2, 3/2)};
				\end{axis}
			\end{tikzpicture} 
		\end{adjustbox}
	}
	\caption{Dual function structure as a function of \(D\). The dual function in (a)-(c) is changed only by the value of \(D\). Observe that the tree phase transitions (indicated by the vertical lines) are independent of the value of \(D\geq0\).}	
	\label{fig:optDualFuncCondCheck}
\end{figure}

Before we discuss the importance of phase transitions on the characterization of strong duality in the IB tree problem, we have the following lemma.
%
\begin{lemma}\label{lem:suffDualAndPTsOPT}
	Assume \(\beta^* \geq 0\) is a solution to the dual problem~\eqref{eq:IBtreeDualProblem}.
    Then there exists a tree phase transition \(\hat \beta\) such that \(d(\hat \beta) = d(\beta^*)\).
\end{lemma}
\begin{proof}
	The proof is given in Appendix~\ref{app:suffDualAndPTsOPTProof}.
\end{proof}
As a result of Lemma~\ref{lem:suffDualAndPTsOPT}, there is no loss of generality in restricting the search for a solution to~\eqref{eq:IBtreeDualProblem} to those \(\beta\) in \(\betaPhaseTransitionSet\), as is done by Algorithm~\ref{alg:subgradOptPhaseTrans}.

\subsection{Characterizing Strong Duality in the IB Tree Problem}

In the previous discussion, we introduced a method to maximize the dual function by exploiting the structure of our problem.
Next, we draw connections between the duality discussion in Section~\ref{sec:RelaxAndDualTrees} and the algorithms discussed in Section~\ref{sec:PTsAndDualOpt}\ref{subsec:IDphaseTransitions}.
Specifically, here we seek to characterize for which values of \(D\) strong duality holds in our problem setting.
This brings us to the following result.
%
\begin{theorem}\label{thm:charaterizeStrongDualityIBtree}
	Strong duality holds in the IB tree problem if and only if \(D = [Q_{Y,\tRoot}]_j\) for some \(j\).
\end{theorem}
\begin{proof}
	The proof is given in Appendix~\ref{app:charaterizeStrongDualityIBtreeProof}.
\end{proof}
The importance of Theorem~\ref{thm:charaterizeStrongDualityIBtree} is that it fully characterizes the values of \(D \geq 0\) for which strong duality holds in the IB tree problem.
This is important, as when strong duality holds in the IB tree problem, a primal optimal solution may be obtained by employing Q-tree search.
Moreover, from this observation, we may also conclude that for those values of \(D \geq 0\) specified by Theorem~\ref{thm:charaterizeStrongDualityIBtree}, Q-tree search finds a tree that is equivalent to that found by solving the primal problem in terms of retained \(X\)- and \(Y\)-information.
It is also important to note that Theorem~\ref{thm:charaterizeStrongDualityIBtree} also establishes that strong duality will fail to hold in \emph{every} IB tree problem, since it is always possible to select a value of \(D \geq 0\) for which there does not exist a tree \(\T_\beta \in\T^\Q\) obtainable by Q-tree search for some \(\beta \geq 0\) such that \(I_Y(\T_\beta) = D\). 

\section{Connections to Integer and Linear Programming} \label{sec:ILPandLPtrees}

The results presented in Section~\ref{sec:RelaxAndDualTrees} showed that it is not generally the case that for any given value of \(D \geq 0\) in~\eqref{eq:IBtreeProb} that there will exist a corresponding value of \(\beta \geq 0\) so that the solutions to~\eqref{eq:HCIBtreeProb} and~\eqref{eq:IBtreeProb} are equivalent.
However, thus far we have only discussed the relationship between the two alternative formulations from the perspective of duality theory, but have not provided any details in regards to how the primal problem~\eqref{eq:HCIBtreeProb} is solved in practice.
As it turns out, the authors of~\cite{larsson2021information,larsson2022linearBC} recently showed that the primal problem may be realized as an integer linear program (ILP), thereby offering another perspective from which we may study the problem~\eqref{eq:HCIBtreeProb} and its relation to~\eqref{eq:IBtreeProb}.
Moreover, equipped with the knowledge that the primal problem~\eqref{eq:HCIBtreeProb} may be realized as an ILP allows for powerful tools from linear programming to be leveraged in our analysis by considering suitable relaxations of the problem~\eqref{eq:HCIBtreeProb}.
Therefore, it is our goal in this section to employ the ILP formulation of the primal problem, and relaxations thereof, to establish a connection between Q-tree search and linear programming, as well as provide a method for selecting \(\beta \geq 0\) as a function of \(D\geq 0\) that leverages strong duality of linear programs.

To this end, the primal problem~\eqref{eq:HCIBtreeProb} can be realized as the ILP
%
\begin{equation}\label{eq:fullILPformulation_opt}
\min \z^{\tp} \Delta_X,
\end{equation}
subject to
%
\begin{align}
\z^{\tp}\Delta_Y &\geq D, \label{eq:fullILPcons1} \\
[\z]_{t'} - [\z]_t &\leq 0, \quad t\in\mathcal{B},~t'\in\chd(t), \label{eq:fullILPcons2} \\
0 \leq [\z]_t &\leq  1,\quad t\in\Nint(\T_\W),\label{eq:fullILPcons3} \\
[\z]_t &\text{ integer} \label{eq:fullILPcons4},
\end{align}
where \(\mathcal{B}\) is the set of interior nodes that have expandable child nodes, \(\Delta_X,\Delta_Y\in\Re^{\lvert\Nint(\T_\W)\rvert}\) have entries assigned according to \([\Delta_X]_t = \Delta I_X(t)\), and \([\Delta_Y]_t = \Delta I_Y(t)\), respectively~\cite{larsson2021information,larsson2022linearBC}.
In~\eqref{eq:fullILPformulation_opt}-\eqref{eq:fullILPcons4}, the constraints~\eqref{eq:fullILPcons2}-\eqref{eq:fullILPcons4} enforce the structure on the vector \(\z\in\{0,1\}^{\lvert \Nint(\T_\W)\rvert}\) to ensure it corresponds to some tree \(\T_{\z}\in\T^\Q\).
The problem~\eqref{eq:fullILPformulation_opt}-\eqref{eq:fullILPcons4} may be written more compactly as
%
\begin{equation}\label{eq:shortFormILP}
\min\{\z^{\tp} \Delta_X : \z^{\tp} \Delta_Y \geq D,~A\z \leq \mathbf{0},~\mathbf{0}\leq \z\leq \mathbf{1},~\z \text{ integer}\},
\end{equation}
where \(A\vec{z}\leq \vec{0}\), with integer matrix \(A\) of dimension \(\left(\sum_{t\in\mathcal B} \lvert \chd(t) \rvert\right) \times \lvert \Nint(\T_\W)\rvert\), represents the constraint~\eqref{eq:fullILPcons2}.
We may draw a connection to the discussion in Section~\ref{sec:RelaxAndDualTrees} by introducing a non-negative weight parameter \(\beta \geq 0\) and by considering a Lagrangian relaxation of~\eqref{eq:shortFormILP} created by incorporating the constraint \(\vec{z}^{\tp} \Delta_Y \geq D\) into the objective to obtain
%
\begin{equation}\label{eq:LagrangianRelaxationOfLPrelax}
\min\{\z^{\tp} \left(\Delta_X - \beta \Delta_Y \right):~A\z \leq \mathbf{0},~\mathbf{0}\leq \z\leq \mathbf{1},~\z \text{ integer}\} + \beta D,
\end{equation}
where one may observe that
%
\begin{align}
&Q(\tRoot;\beta) = \nonumber \\
&\min\{\z^{\tp} \left(\Delta_X - \beta \Delta_Y \right):~A\z \leq \mathbf{0},~\mathbf{0}\leq \z\leq \mathbf{1},~\z \text{ integer} \}, \label{eq:QasILP} 
\end{align}
since \(\T^\Q = \{\z: A\z \leq \mathbf{0},~\mathbf{0}\leq \z\leq \mathbf{1}\}\), and for any \(\z\in\{\z: A\z \leq \mathbf{0},~\mathbf{0}\leq \z\leq \mathbf{1}\}\) we have \(\z^{\tp} \left(\Delta_X - \beta \Delta_Y \right) = I_X(\T_{\z}) - \beta I_Y(\T_\z)\), where \(\T_\z \in\T^\Q\) is the tree corresponding to the vector \(\z\)~\cite{larsson2021information,larsson2022linearBC}.
As a result,~\eqref{eq:QasILP} shows that the value of \(Q(\tRoot;\beta)\) may be obtained by solving an ILP.
However, much more can be deduced from the relation~\eqref{eq:QasILP} by leveraging strong duality of linear programming~\cite[ch.~4]{bertsimas1997introduction}.
To this end, we note that the problem~\eqref{eq:LagrangianRelaxationOfLPrelax} (or, equivalently~\eqref{eq:QasILP}) would be a linear program were it not for the complicating constraint that the vector \(\vec{z}\) be integer.
One may work around this complicating fact to obtain a linear program by considering a relaxation of~\eqref{eq:LagrangianRelaxationOfLPrelax} formed by removing the integer constraint on \(\vec{z}\).
It is important to note that, since the Q-tree search algorithm can solve the optimization problem~\eqref{eq:LagrangianRelaxationOfLPrelax} for a given value of \(\beta \geq 0\) over the space of trees, such a relaxation does not provide much utility.
Instead, we will investigate whether or not~\eqref{eq:LagrangianRelaxationOfLPrelax} has the \emph{integrality property}~\cite{Geoffrion2010}, thereby allowing the complicating integer constraint in problem~\eqref{eq:LagrangianRelaxationOfLPrelax} to be removed without inducing a loss in the resulting relaxation.
In other words, if the feasible set \(\{\z: A\z \leq \mathbf{0},~\mathbf{0}\leq \z\leq \mathbf{1}\}\) of~\eqref{eq:LagrangianRelaxationOfLPrelax} has the integrality property, then the problem~\eqref{eq:LagrangianRelaxationOfLPrelax} may be solved equivalently as a regular linear program.
This brings us to the following definition.
%
\begin{definition}[\hspace{-1sp}\cite{wolsey1999integer}]\label{def:polyhebra}
	A \emph{Polyhedron} \(P\subseteq \Re^n\) is a set of the form \(P = \{\vec{z} : B \vec{z} \leq \vec{b}\}\), where \(B \in \Re^{m\times n}\) and \(\vec{b}\in \Re^m\).
	Moreover, \(\vec{z}\in P\) is an \emph{extreme point} of \(P\) if there does not exist \(\vec{z}_1,\vec{z}_2\in P\), \(\vec{z}_1\neq\vec{z}_2\), such that \(\vec{z} = \frac{1}{2}\vec{z}_1 + \frac{1}{2}\vec{z}_2\).
\end{definition}
Note that the set \(\{\z: A\z \leq \vec{0},~\vec{0}\leq \z\leq \vec{1}\}\) can be written in the form of Definition~\ref{def:polyhebra} by taking \(B = \left[A^\tp, I_{N}, -I_{N}\right]^\tp\) and \(\vec{b} = \left[\vec{0}^\tp,\vec{0}^\tp,\vec{1}^\tp\right]^\tp\), where \(N = \lvert \Nint(\T_\W)\rvert\).
Moreover, by writing the matrix \(B\) for our problem in this manner allows us to conclude that \(B\) has full column rank. 
With this in mind, we next define what it means for a polyhedron to be integral.
%
\begin{definition}[\hspace{-1sp}{\cite[p.~536]{wolsey1999integer}}]\label{def:integralPolyhedra}
	Let \(B\in\Re^{m\times n}\).
	A non-empty polyhedron \(P = \{\z\in\Re^n : B\z \leq \vec{b}\}\) with \(\text{rank}(B) = n\) is integral if all of its extreme points are integral.\footnote{Strictly speaking, Definition~\ref{def:integralPolyhedra} is not the \emph{definition} of integer polyhedra according to~\cite{wolsey1999integer}. However, it is an equivalent characterization of integer polyhedra that suffices for our work. See~\cite[p.~536]{wolsey1999integer} for more details.}
\end{definition}
In the language of our problem, if the polyhedron \(\{\z: A\z \leq \vec{0},~\vec{0}\leq \z\leq\vec{1}\}\) is integral, then the problem~\eqref{eq:LagrangianRelaxationOfLPrelax} may be equivalently solved as a linear program.
However, while it is intuitive that the integral property of a polyhedron depends on the inequalities that define the set, it is not clear exactly how one may establish that a polyhedron is integral from the provided constraint information (i.e., \(B\) and \(\vec{b}\)).
For this, we require the following definition. 
%
\begin{definition}[\hspace{-1sp}{\cite[p.~540]{wolsey1999integer}}]
	Let \(A\) be an integer matrix.
	Then A is \emph{totally unimodular} if the determinant of each square submatrix of \(A\) is either \(-1,~0\) or \(1\).
\end{definition}
The following theorem makes the connection between total unimodularity and integral polyhedra precise.
%
\begin{theorem}[\hspace{-1sp}{\cite[p.~541]{wolsey1999integer}}]\label{thm:TUMandIntegralPolyhedra}
	Let \(A\) be a totally unimodular matrix.
	Then the polyhedron \(\{\z: A\z \leq \vec{0},~\vec{0}\leq \z \leq \vec{1}\}\) is integral so long as it is not empty.
\end{theorem}
The significance of Theorem~\ref{thm:TUMandIntegralPolyhedra} is that it provides a means to determine whether or not a polyhedron of the form \(\{\z: A\z\leq \vec{0},~\vec{0}\leq\z\leq\vec{1}\}\) is integral directly from the constraint matrix \(A\).
The challenge is then to determine whether or not the matrix \(A\) in problem~\eqref{eq:LagrangianRelaxationOfLPrelax} is totally unimodular.

Given the importance of total unimodularity in establishing the integrality property of integer programs, it is unsurprising that a number of conditions have been put forth that allows one to establish that a matrix is totally unimodular (see, for example,~\cite{wolsey1999integer} and~\cite{schrijver1998theory}).
One criterion, which will be employed in what follows, is summarized by the following theorem.
%
\begin{theorem}[\hspace{-1sp}{\cite[p.~542]{wolsey1999integer}}]\label{thm:necSufTUMConditions}
	Let \(A\) be an \(m\times n\) integer matrix.
	If for every collection of columns \(J\subseteq\{1,\ldots,n\}\) of \(A\) there exists a partition \(J_1\), \(J_2\) of \(J\) such that
	\begin{equation*}
		\left| \sum_{j\in J_1} [A]_{ij} - \sum_{j\in J_2} [A]_{ij} \right| \leq 1,
	\end{equation*}
	for all \(i=1,\ldots,m\), then \(A\) is totally unimodular.
\end{theorem}
We are now in a position to show that the polyhedron \(\{\z : A\z \leq \vec{0},~\vec{0}\leq\z\leq\vec{1}\}\) in problem~\eqref{eq:LagrangianRelaxationOfLPrelax} is integral.
%
\begin{proposition}\label{prop:TUMIBtrees}
	Let \(A\) be the integer matrix in problem~\eqref{eq:LagrangianRelaxationOfLPrelax}.
	Then \(A\) is totally unimodular.
\end{proposition}
%
\begin{proof}
	The proof is given in Appendix~\ref{app:proofOfTUMProp}.
\end{proof}
As a result of Proposition~\ref{prop:TUMIBtrees}, we may conclude that the polyhedron in problem~\eqref{eq:LagrangianRelaxationOfLPrelax} is integral.
In turn, this implies that the integer program~\eqref{eq:LagrangianRelaxationOfLPrelax} may be solved by the linear program obtained when the integer constraint is removed.
More explicitly, Proposition~\ref{prop:TUMIBtrees} implies that~\eqref{eq:LagrangianRelaxationOfLPrelax} is equivalent to the \emph{linear program}
%
\begin{equation}\label{eq:LPformOfQtreeSearchProblem}
\min\{\z^\tp (\Delta_X - \beta \Delta_Y) : A\z \leq \vec{0},~\vec{0}\leq\z\leq\vec{1}\} + \beta D.
\end{equation}
We conclude that the problem~\eqref{eq:IBtreeProb} for which Q-tree search finds an optimal tree solution for a given value of \(\beta \geq 0\) can be equivalently realized as a linear program according to~\eqref{eq:LPformOfQtreeSearchProblem}.
Consequently, in order to obtain a solution to~\eqref{eq:IBtreeProb} for a given value of \(\beta \geq 0\), one may either employ Q-tree search or solve the linear program~\eqref{eq:LPformOfQtreeSearchProblem}.
In light of the above observations, we may exploit the strong duality of linear programs to develop an alternate method to Algorithm~\ref{alg:subgradOptPhaseTrans} for selecting a value of \(\beta \geq 0\) as a function of \(D \geq 0\).

\subsection{Leveraging LP Duality to Select \(\beta\) as a function of \(D\)}\label{subsec:}
While the observation that problem~\eqref{eq:LagrangianRelaxationOfLPrelax} may be equivalently solved as a linear program is useful on its own, it also tacitly provides a method by which one may select \(\beta \geq 0\) as a function of \(D \geq 0\) by leveraging the strong duality property of linear programs.
In more detail, the linear programming relaxation of the primal problem~\eqref{eq:shortFormILP} for a given value of \(D \geq 0\) is given by
%
\begin{equation}\label{eq:LPrelaxOfILP}
\bar v(D) = \min\{\z^\tp \Delta_X : \z^\tp \Delta_Y\geq D,~A\z\leq \vec{0},~\vec{0}\leq\z\leq\vec{1}\}.
\end{equation}
Note that~\eqref{eq:LPrelaxOfILP} is nothing more than~\eqref{eq:shortFormILP} with the integer constraints removed.
It is also important to note that the solution to~\eqref{eq:LPrelaxOfILP} need not be integer.
However, problems~\eqref{eq:LPformOfQtreeSearchProblem} and~\eqref{eq:LPrelaxOfILP} are related via the dual function of~\eqref{eq:LPrelaxOfILP}.
Specifically, the dual of the linear programming relaxation~\eqref{eq:LPrelaxOfILP} with respect to the constraint \(\z^\tp \Delta_Y\geq D\) is given by 
%
\begin{equation}\label{eq:dualFunctionOfLPRelax}
\bar d(\beta) = \min\{\z^{\tp}(\Delta_X - \beta \Delta_Y): A\z\leq\vec{0},~\vec{0}\leq\z\leq\vec{1}\} + \beta D,
\end{equation}
where \(\beta\geq 0\) is a dual variable.
The dual problem of~\eqref{eq:LPrelaxOfILP} is then 
%
\begin{equation}\label{eq:LPrelaxationDualProblem}
\max \{\bar d(\beta) : \beta \geq 0\}.
\end{equation}
Letting \(\beta^*\in\argmax\{\bar d(\beta) : \beta \geq 0\}\), it follows from the strong duality property of linear programs~\cite[ch.~4]{bertsimas1997introduction} that \(\bar d(\beta^*) = \bar v(D)\).
As a result, we see that a value of \(\beta^* \geq 0\) for a given value of \(D \geq 0\) can be obtained as the dual variable corresponding to the constraint \(\z^\tp \Delta_Y \geq D\) in the linear program~\eqref{eq:LPrelaxOfILP}.
It is important to note that, via the integrality property of the problem~\eqref{eq:LagrangianRelaxationOfLPrelax} (which is equivalent to~\eqref{eq:dualFunctionIBTrees}) and Lemma~\ref{lem:suffDualAndPTsOPT}, it follows that \(\beta^* \geq 0\) obtained by maximizing \(\bar d(\beta)\) over \(\beta \geq 0\) will satisfy \(d(\beta^*) = d(\hat \beta)\) for some tree phase transition \(\hat \beta\in \betaPhaseTransitionSet\).
Consequently, for a given value of \(D \geq 0\), we may view the linear programming dual approach delineated above and Algorithm~\ref{alg:subgradOptPhaseTrans} as equivalent in terms of finding a setting of \(\beta \geq 0\) that maximizes the dual function.
We remark, however, that Algorithm~\ref{alg:subgradOptPhaseTrans} returns the collection of \emph{all} tree phase transitions, whereas the LP dual approach discussed above returns only a single value of \(\beta \geq 0\).
Moreover, the approach employed by Algorithm~\ref{alg:subgradOptPhaseTrans} leverages the mechanism of Q-tree search to determine the tree phase transitions, and does not require solving a (potentially large) linear programming problem.
%

\section{Numerical Example \& Discussion} \label{sec:resultsDiscuss}

\begin{figure*}
    \subfloat[\(D / I(X;Y) =  0.59\)\label{fig:dualShape1}]{\includegraphics[width=0.3\textwidth]{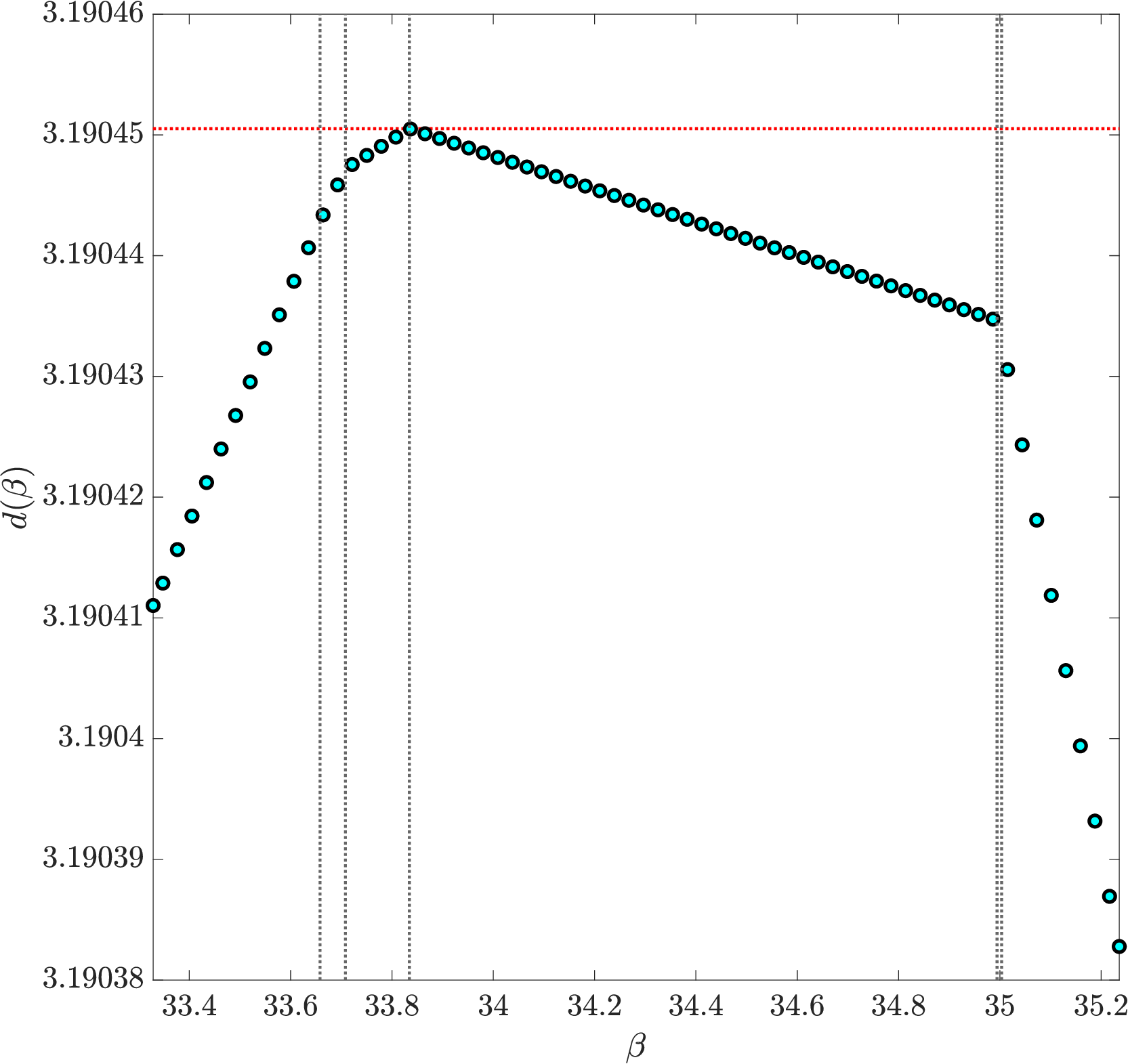}}
    \hfil
    \subfloat[\(D / I(X;Y) = 0.69\)\label{fig:dualShape2}]{\includegraphics[width=0.3\textwidth]{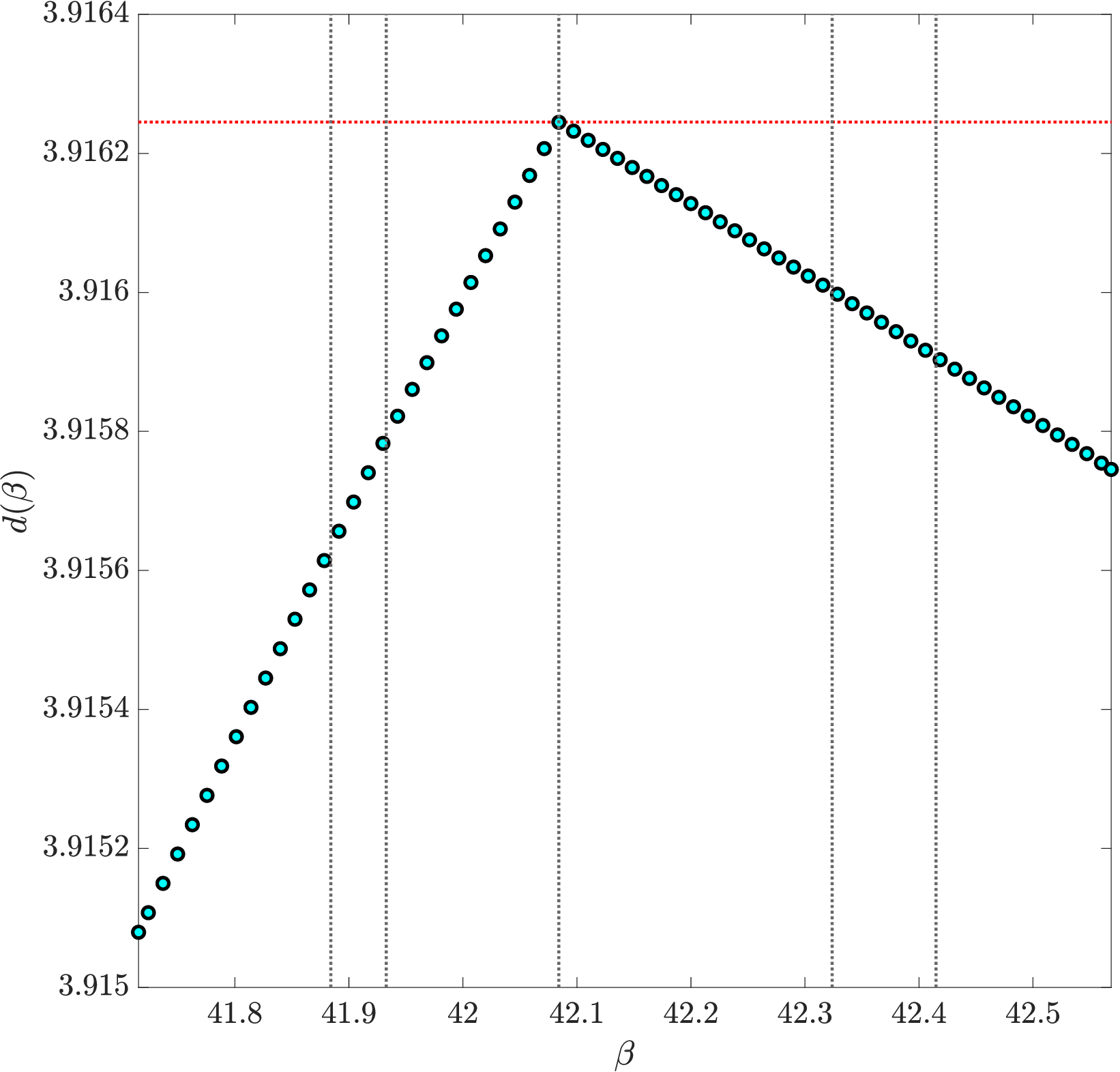}}
    \hfil
    \subfloat[\(D / I(X;Y) = 0.74\)\label{fig:dualShape3}]{\includegraphics[width=0.3\textwidth]{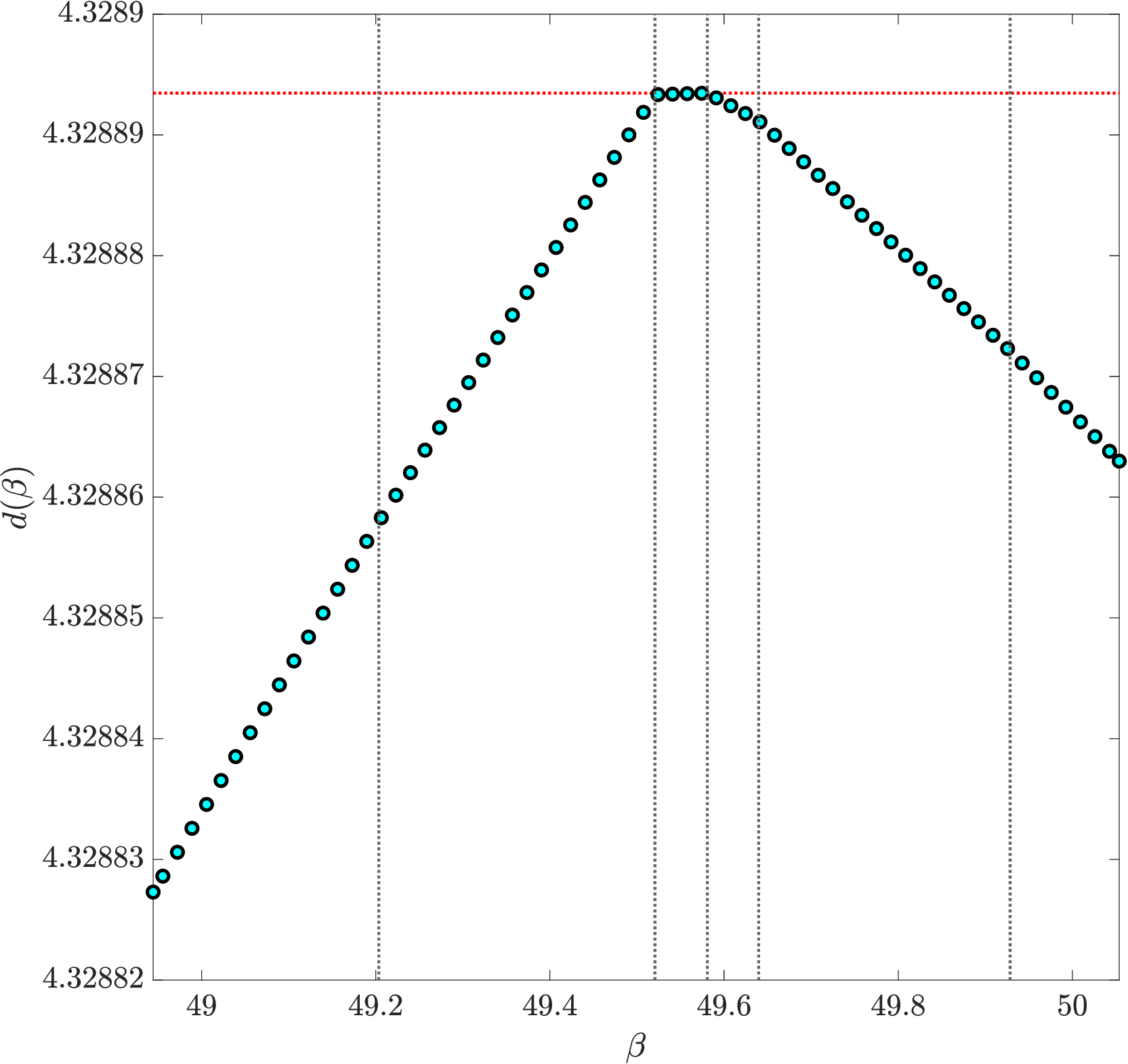}}
    \caption{Simulation results showing the shape of dual function in the vicinity of its maximal value versus \(\beta\) for various \(D\). In each figure, the dual function is shown using two methods: the Q-function method (black dots) and the phase-transition method (cyan dots). The maximal dual function value is obtained from the linear programming approach discussed in Section~\ref{sec:ILPandLPtrees} (dashed horizontal red line), alongside the computed tree phase transitions (vertical lines).}
    \label{fig:dualFunctionShapeLocalOPT}
\end{figure*}

%
In this section, we discuss simulation results obtained from implementation of the theoretical developments throughout the paper.
To this end, we consider generating hierarchical, multi-resolution, tree abstractions of the \(128 \times 128\) environment show in Figure~\ref{fig:std128Env}.
We assume that \(p(x,y) = p(y|x) p(x)\), where the relevant information in the environment is represented by \(Y: \Omega \to \{0,1\}\) where \(p(Y = 1|x)\) is given by the cell intensity shown in Figure~\ref{fig:std128Env}, and \(p(x)\) is assumed to be uniform over finest-resolution cells.
%

\begin{figure}[tbh]
    \centering
    \includegraphics[width=0.85\columnwidth]{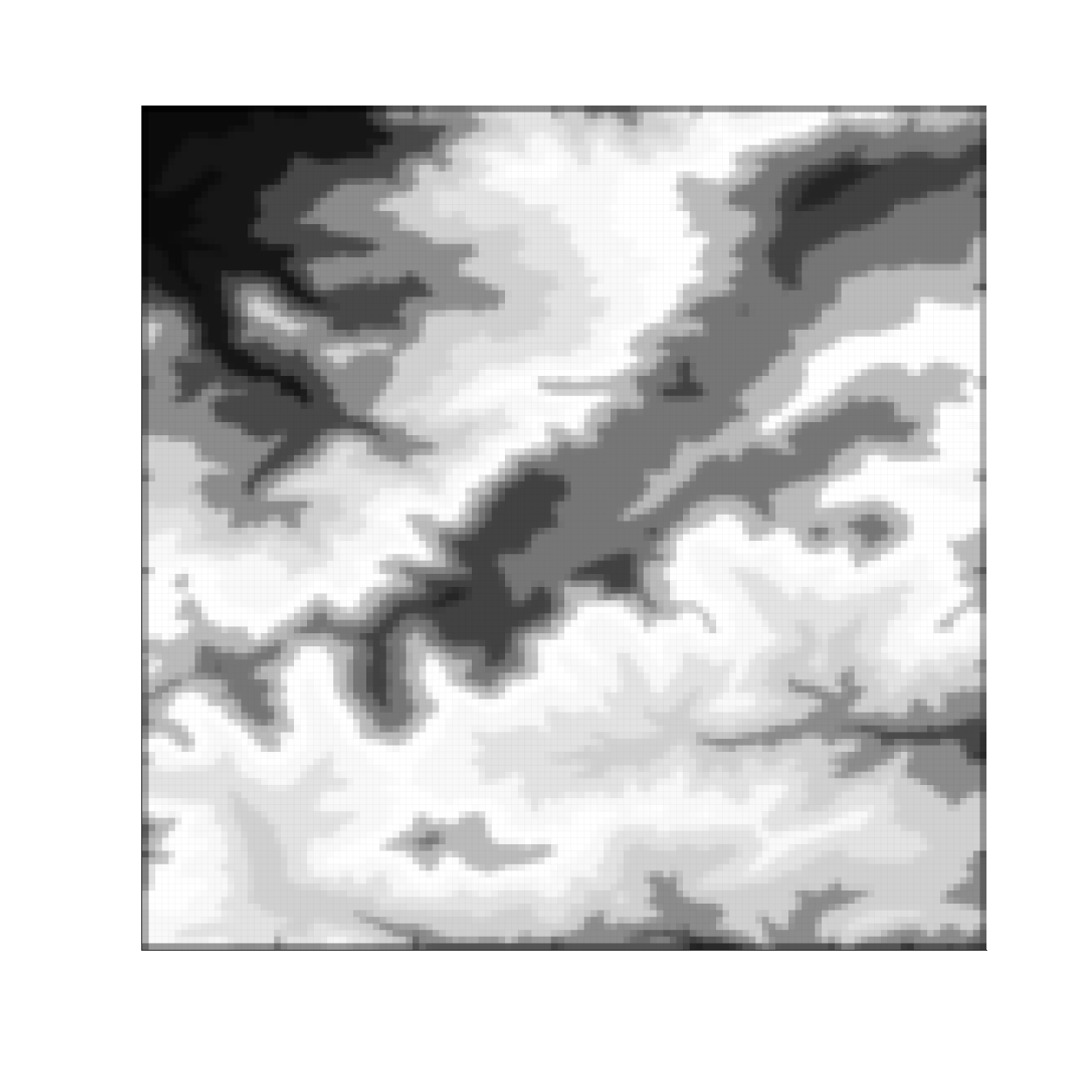}
    \caption{Finest-resolution \(128 \times 128\) enviroment.}
    \label{fig:std128Env}
\end{figure}
In Figure~\ref{fig:dualOptValueVariousMethods} we show the dual function in the vicinity of its maximial value for various settings of \(D\). 
It is important to note that there are three distinct methods employed to generate the results shown in Figure~\ref{fig:dualFunctionShapeLocalOPT}.
Specifically, the results in Figure~\ref{fig:dualFunctionShapeLocalOPT} were created using: (i) the Q-function method, which employs~\eqref{eq:dualFunctionQ} to generate the value of \(d(\beta)\), (ii) the phase transition method, which utilizes the phase transitions computed from Algorithm~\ref{alg:phaseTransitionAlgorithm} and~\eqref{eq:dualFunctionInTermsOfPTalgVars} to determine \(d(\beta)\), and (iii) the linear programming approach discussed in Section~\ref{sec:ILPandLPtrees} which uses the total unimodularity property to obtain the maximal value of the dual problem~\eqref{eq:IBtreeDualProblem}.
Even more importantly, Figure~\ref{fig:dualFunctionShapeLocalOPT} shows the piecewise linear character of the dual function, and how the dual function is changed at the tree phase transitions.
Namely, we see in the results displayed in Figure~\ref{fig:dualFunctionShapeLocalOPT} how the slope of \(d(\beta)\) is changed at the tree phase transitions, which are shown by the vertical lines in the figure (compare Figures~\ref{fig:dualFunctionVisualization} and~\ref{fig:dualFunctionShapeLocalOPT}).
%

\begin{figure}[bth]
    \centering
    \includegraphics[width=0.8\columnwidth]{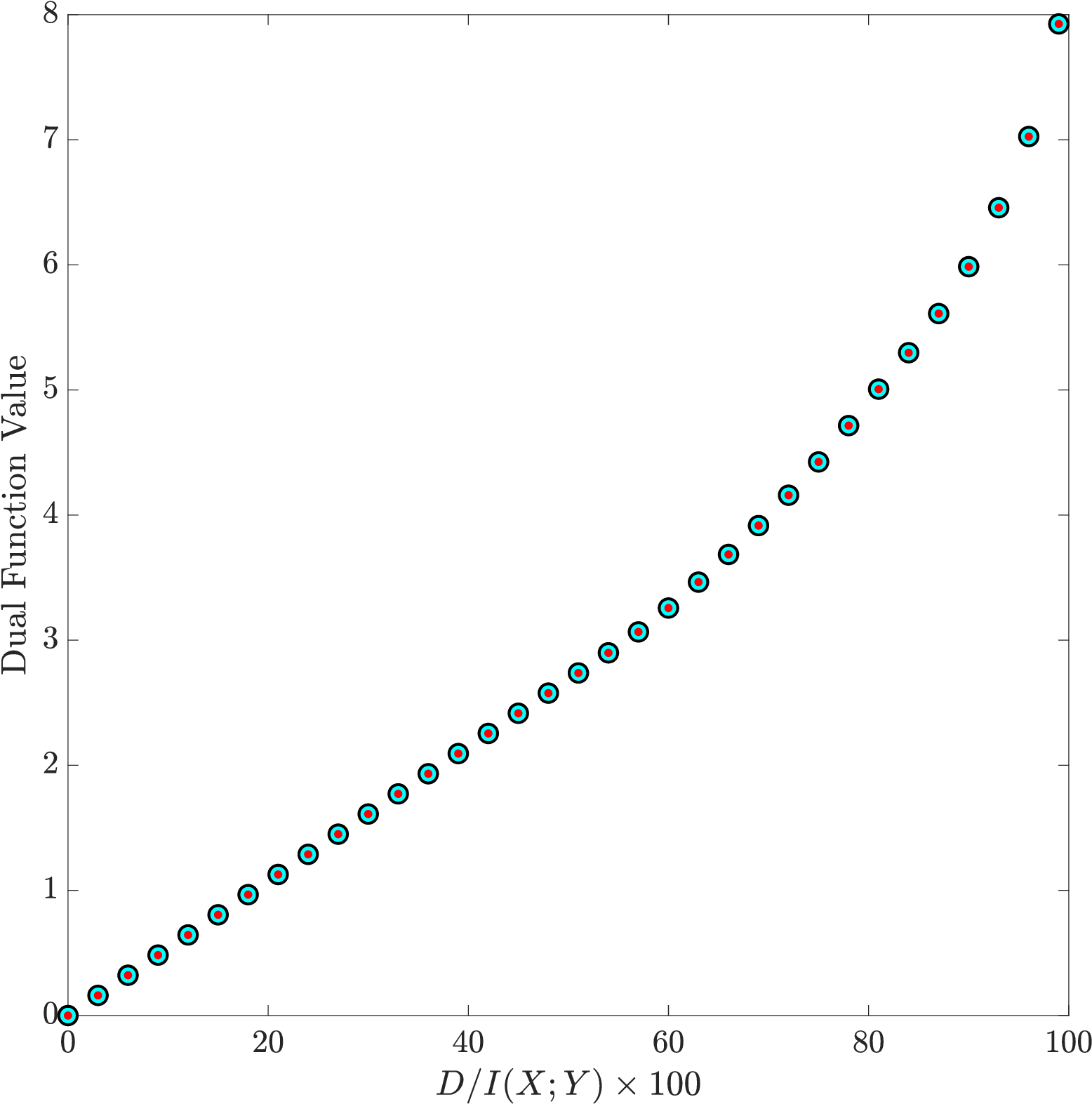}
    \caption{Optimal dual function value versus \(D\) obtained by employing three distinct approaches: (i) linear programming relaxation (red dots), (ii) phase transition method (i.e., Algorithm~\ref{alg:subgradOptPhaseTrans}) (cyan dots), (iii) Q-tree search (black squares).}
    \label{fig:dualOptValueVariousMethods}
\end{figure}
Moreover, we also see that employing Q-tree search to obtain a tree that attains the dual optimal value may result in a tree that is not primal feasible.
To more clearly understand why this is the case, we recall that, if \(\beta_{\text{lb}}\) and \(\beta_{\text{ub}}\) are two distinct tree phase transitions satisfying \(\beta_{\text{lb}} < \beta_{\text{ub}}\), then for all \(\beta \in (\beta_{\text{lb}}, \beta_{\text{ub}}]\) the tree solution to~\eqref{eq:IBtreeProb} returned by Q-tree search will be unchanged and determined by \(\beta_{\text{ub}}\).
Furthermore, if we let the corresponding Q-tree search solution for the interval \((\beta_{\text{lb}}, \beta_{\text{ub}}]\) be denoted \(\T_{\beta_{\text{ub}}}\in\T^\Q\), it then follows that \(D - I_Y(\T_{\beta_{\text{ub}}}) \geq 0\) or \(D - I_Y(\T_{\beta_{\text{ub}}}) < 0\) for the tree \(\T_{\beta_{\text{ub}}}\).
Thus, considering the search process of Algorithm~\ref{alg:subgradOptPhaseTrans}, we see that in cases when there does not exist a tree obtainable by Q-tree search with exactly \(D\) units of information, we have \(D - I_Y(\T_{\beta_{\text{ub}}}) > 0\) or \(D - I_Y(\T_{\beta_{\text{ub}}}) < 0\).
Of these two alternatives, Q-tree search will always return the tree for which \(D - I_Y(\T_{\beta_{\text{ub}}}) > 0\), as the maximal value of the dual function will be attained at the tree phase transition \(\beta_{\text{ub}}\).
As a result, the Q-tree search returned tree will not be primal feasible.
However, should a primal feasible tree be desired, one may create one by running Q-tree search with the trade-off parameter set to \(\beta_{\text{ub}} + \varepsilon\) for any \(\varepsilon > 0\), as this will ensure that a tree phase transition occurs resulting with \(D - I_Y(\T_{\beta_{\text{ub}} + \varepsilon}) < 0\).
Moreover, it follows from~\cite{Geoffrion2010} that the sub-optimality of the primal feasible tree \(\T_{\beta_{\text{ub}} + \varepsilon} \in \T^\Q\) is given by \(0 \leq I_X(\T_{\beta_{\text{ub}} + \varepsilon}) - v(D) \leq (\beta_{\text{ub}} + \varepsilon)(I_Y(\T_{\beta_{\text{ub}} + \varepsilon}) - D)\).
The case when Q-tree search does not return a primal feasible solution may be seen in Figures~\ref{fig:dualShape1} and~\ref{fig:dualShape2}, where the slope of the dual function switches from positive to negative on either side of the maximum (i.e., some tree phase transition). 
In light of the above discussion, if Q-tree search returns a tree that is not primal feasible, a primal feasible tree may be easily obtained by triggering a subsequent tree phase transition, which is accomplished by perturbing the dual variable by a small number \(\varepsilon > 0\).
Furthermore, we show in Figure~\ref{fig:dualOptValueVariousMethods} the dual optimal value as a function of \(D\) as determined by the approaches discussed in Sections~\ref{sec:PTsAndDualOpt} and~\ref{sec:ILPandLPtrees}.
Specifically, in Figure~\ref{fig:dualOptValueVariousMethods} we show the dual optimal value using: (i) the linear programming relaxation, which utilizes total unimodularity property, the phase transition method which employs Algorithm~\ref{alg:subgradOptPhaseTrans} to find a setting of the dual variable that maximizes the dual function, and (iii) the Q-function method~\eqref{eq:dualFunctionQ} when \(\beta\) is selected as the dual variable obtained from the linear programming relaxation method.
We again see very good agreement between the approaches.
%

\begin{figure}[bth]
    \centering
    \includegraphics[width=0.8\columnwidth]{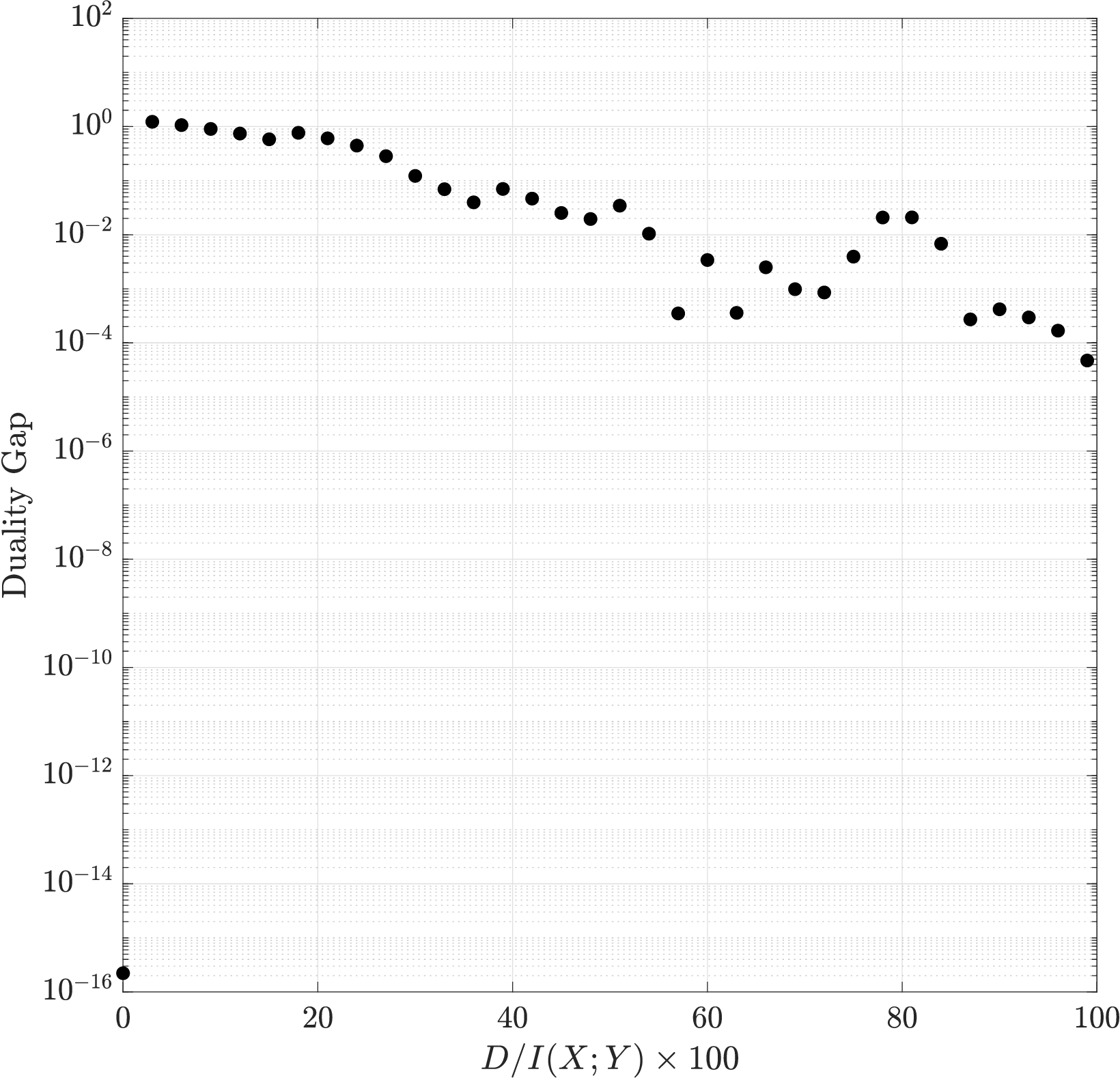}
    \caption{Duality gap as a function of \(D\).}
    \label{fig:dualGapVsD}
\end{figure}
Lastly, in Figure~\ref{fig:dualGapVsD}, we show the duality gap as the setting of the hard-constraint \(D\) is changed.
Recall that the duality gap is the difference between the primal and optimal values; that is, \(v(D) - d(\beta^*)\) where \(\beta^*\) is a solution to the dual problem~\eqref{eq:IBtreeDualProblem}.
In the results shown in Figure~\ref{fig:dualGapVsD} there is only a single point where strong duality holds, which occurs when \(D = 0\).
This is perhaps unsurprising, since the root tree (i.e., the tree aggregating all finest-resolution cells to a single node) does not contain any relevant information, and thus there exists a tree obtainable by employing Algorithm~\ref{alg:phaseTransitionAlgorithm} and~\ref{alg:subgradOptPhaseTrans} and so the conditions of Theorem~\ref{thm:charaterizeStrongDualityIBtree} are satisfied.
As shown in Figure~\ref{fig:dualGapVsD}, in all other cases there is a duality gap which occurs since there does not exist a tree obtainable by Q-tree search which has exactly \(D\) units of relevant information.
In other words, strong duality fails to hold when there does not exist a tree solution to~\eqref{eq:IBtreeProb} that retains precisely \(D\) units of information regarding the relevant variable \(Y\).
In this light we may think of the utility of Algorithm~\ref{alg:phaseTransitionAlgorithm} as not only enabling efficient methods to find a setting of the dual variable that solves the dual problem (e.g., via Algorithm~\ref{alg:subgradOptPhaseTrans}), but also to specify for exactly which values of \(D\) the Q-tree search method can be used to generate a primal optimal solution.
%

\section{Conclusions}\label{sec:conclusion}

In this paper, we have considered the problem of establishing a formal connection between two information-theoretic tree abstraction problems that leverage the information-bottleneck principle.
More specifically, we consider the problem of determining the relationship between the hard- and soft-constrainted tree-search formulations that originally appeared in~\cite{larsson2021information} and~\cite{larsson2020q}, respectively.
To accomplish our goal, we leverage concepts from Lagrangian relaxation, duality theory, and integer programming.
Our accomplishments in this paper are two-fold: (1) we establish that the soft-constrained and hard-constrained, discrete, tree-search formulations are generally not equivalent, and (2) we investigate and develop methods that allow the selection of the weight parameter of the soft-constrained formulation as a function of the primal constraint.
To accomplish our goals, we delineate how the weight parameter of the soft-constrained problem may be viewed as a dual variable of the hard-constrained problem by establishing a bridge between the dual function and the Q-function from Q-tree search~\cite{larsson2020q}.
We then show how the structure of our problem may be exploited to develop an algorithm which allows the dual problem to be tractably solved as a function of the primal constraint that operates by leveraging the properties of the Q-function and the mechanism of Q-tree search.
An algorithm for solving the dual problem is proposed and utilized to characterize the conditions under which strong duality will hold in the discrete, tree-abstraction problem.
Furthermore, we investigate the integer linear programming formulation of the hard-constrained problem~\cite{larsson2021information,larsson2021informationB} and develop an alternate approach towards the selection of the dual variable via strong duality of linear programs.
To this end, we show that the constraint matrix characterizing the set of all valid, multi-resolution hierarchical tree abstractions is totally unimodular, and delineate how this result may be utilized to draw a connection between the weight parameter in the soft-constrained tree-search problem and the dual variable in the linear programming relaxation of the hard-constrained formulation.
Empirical results are provided that corroborate the theoretical developments, and the importance of the work for abstraction design in resource-limited autonomous systems was discussed.
%

\begin{appendices}

\section{Proof of Proposition~\ref{prop:QfunctionAndObjIBTree}}\label{app:QfunctionAndObjIBTreePrf}	

Let \(\beta \geq 0\) be given, and assume \(\T_{\beta}\in\T^\Q\) is a tree that attains the minimum in~\eqref{eq:IBtreeProb}; that is, 
\begin{equation*}
	\T_\beta \in \argmin\{I_X(\T) - \beta I_Y(\T): \T\in\T^\Q\}.
\end{equation*}
Next, let \(\Tilde \T \in \T^\Q\) be the tree for which \(Q(\tRoot;\beta) = \sum_{s\in\Nint(\tilde \T_{(\tRoot)})}\Delta I_X(s) - \beta \Delta I_Y(s)\), which is known to exist from the second part of Lemma~\ref{lem:existenceOfTreeAndQfunction}.
Importantly, notice that \(\Nint(\tilde \T) = \Nint(\tilde \T_{(\tRoot)})\) as all nodes in the tree \(\tilde \T\) are descendant from the root node.
As a result, \(Q(\tRoot;\beta) = I_X(\tilde \T) - \beta I_Y(\tilde \T)\) for some \(\tilde \T\in\T^\Q\), and thus, from the definition of \(\T_\beta\), we have
\begin{equation*}
	I_X(\T_\beta) - \beta I_Y(\T_\beta) \leq I_X(\tilde \T) - \beta I_Y(\tilde \T) = Q(\tRoot;\beta).
\end{equation*}
However, from the first portion of Lemma~\ref{lem:existenceOfTreeAndQfunction}, it follows that \(Q(\tRoot;\beta) \leq \sum_{s\in \Nint(\T_{(\tRoot)})} \Delta I_X(s) - \beta \Delta I_Y(s) = I_X(\T) - \beta I_Y(\T)\) for all \(\T\in\T^\Q\).
Consequently, since \(\T_\beta \in \T^\Q\), we have 
\begin{align*}
	Q(\tRoot;\beta)& \leq \sum_{s \in \Nint(\T_{\beta(\tRoot)})}\Delta I_X(s) - \beta \Delta I_Y(s), \\
	&=  I_X(\T_\beta)- \beta I_Y(\T_\beta).
\end{align*}
As a result, we obtain
\begin{align*}
	Q(\tRoot;\beta) &\leq I_X(\T_\beta) - \beta I_Y(\T_\beta),\\
	&\leq I_X(\tilde \T) - \beta I_Y(\tilde \T) = Q(\tRoot;\beta),
\end{align*}
and therefore, \(Q(\tRoot;\beta) = I_X(\T_\beta) - \beta I_Y(\T_\beta)\), or, equivalently, 
\begin{equation*}
	Q(\tRoot;\beta) = \min\{I_X(\T) - \beta I_Y(\T): \T\in\T^\Q\},
\end{equation*}
which establishes the result. \hfill \(\qedsymbol\)

\section{Proof of Proposition~\ref{prop:QfunctionMonotone}}\label{app:QfunctionMonotonePrf}

The proof is given by induction.
Consider any node \(t \in \N_{\ell - 1}(\T_\W)\), and assume \(\beta_1 \geq \beta_2\).
Notice that, since \(\beta_1 \geq \beta_2\), we have \(\Delta I_X(t) - \beta_1 \Delta I_Y(t) \leq \Delta I_X(t) - \beta_2 \Delta I_Y(t)\).
Consequently, we have \(\min\{\Delta I_X(t) - \beta_1 \Delta I_Y(t),~0\} \leq \min\{\Delta I_X(t) - \beta_2 \Delta I_Y(t),~0\}\), and therefore
\begin{align*}
	Q(t;\beta_1) = &\min\{\Delta I_X(t) - \beta_1 \Delta I_Y(t),~0\}, \\
	&\leq \min\{\Delta I_X(t) - \beta_2 \Delta I_Y(t),~0\} = Q(t;\beta_2),
\end{align*}
where equality follows from the definition of the Q-function~\eqref{eq:QfunctionDef} for the node \(t\), recalling that nodes at depth \(\ell\) have a Q-function value equal to zero.
The above establishes \(Q(t;\beta_1) \leq Q(t;\beta_2)\) for all \(t\in\N_{\ell-1}(\T_\W)\).

\vspace{5pt}

\noindent 
Assume that the result holds for all \(t\in \N_{k+1}(\T_\W)\), \(0 \leq k \leq \ell - 2\), and consider any \(t\in \N_{k}(\T_\W)\).
Then, since \(\beta_1 \geq \beta_2\) we have that \(\Delta I_X(t) - \beta_1 \Delta I_Y(t) \leq \Delta I_X(t) - \beta_2 \Delta I_Y(t)\), which, together with the induction hypothesis, yields
\begin{align*}
	\Delta I_X(t) - \beta_1 &\Delta I_Y(t) + \sum_{t'\in\chd(t)} Q(t';\beta_1) \\
	&\leq \Delta I_X(t) - \beta_2 \Delta I_Y(t) + \sum_{t'\in\chd(t)} Q(t';\beta_1), \\
	&\leq \Delta I_X(t) - \beta_2 \Delta I_Y(t) + \sum_{t'\in\chd(t)} Q(t';\beta_2),
\end{align*}
where the final inequality utilizes the induction hypothesis as \(\chd(t) \subseteq \N_{k+1}(\T_\W)\).
As a result, we obtain
\begin{align*}
	Q(t;\beta_1) &= \min\{\Delta I_X(t) - \beta_1 \Delta I_Y(t) + \sum_{t'\in\chd(t)} Q(t';\beta_1),~0\},\\
	&\leq \min\{\Delta I_X(t) - \beta_2 \Delta I_Y(t) + \sum_{t'\in\chd(t)} Q(t';\beta_2),~0\},\\
	&= Q(t;\beta_2).
\end{align*}
Therefore, if \(\beta_1 \geq \beta_2\) then \(Q(t;\beta_1) \leq Q(t;\beta_2)\). \hfill \(\qedsymbol\)
%

\section{Proof of Lemma~\ref{lem:suffDualAndPTsOPT}}\label{app:suffDualAndPTsOPTProof}	

Let \(\beta^* \geq 0\) be a solution to the dual problem~\eqref{eq:IBtreeDualProblem}.
Now define \(\beta_{\textrm{min}} = \min(\betaPhaseTransitionSet)\) and \(\beta_{\textrm{max}} = \max(\betaPhaseTransitionSet)\), and consider three cases: \(0\leq \beta^* \leq \beta_{\textrm{min}}\), \(\beta_{\textrm{max}} \leq \beta^*\), and \(\beta_{\textrm{min}}\leq\beta^*\leq\beta_{\textrm{max}}\).

\vspace{5pt}

\noindent
Case I: \(0\leq \beta^* \leq \beta_{\textrm{min}}\). 
As \(\betacrq(\tRoot) = \beta_{\textrm{min}}\), it follows that for \(0\leq \beta^* \leq \beta_{\textrm{min}}\), \(Q(\tRoot;\beta^*) = 0\) and thus, \(d(\beta^*) = \beta^* D\) according to~\eqref{eq:dualFunctionQ}.
As a result, since \(D \geq 0\), \(d(\beta)\) is monotone increasing on \([0,\beta_{\textrm{min}}]\).
Consequently, \(d(\beta^*) \leq d(\beta_{\textrm{min}})\) and because \(d(\beta^*) \geq d(\beta)\) for all \(\beta \geq 0\), the above yields \(d(\beta^*) = d(\beta_{\textrm{min}})\).
Since \(\beta_{\textrm{min}} = \betacrq(\tRoot)\), and \(\betacrq(\tRoot)\) is a tree phase transition, it follows that there exists a tree phase transition \(\hat \beta = \betacrq(\tRoot)\) such that \(d(\hat \beta) = d(\beta^*)\).

\vspace{8pt}

\noindent
Case II: \(\beta_{\textrm{max}} \leq \beta^*\).
For \(\beta_{\textrm{max}} \leq \beta^*\), it follows from~\eqref{eq:dualIBtreeFunctionWithDeltaInfo} that the dual function linear: \(d(\beta) = I_X(\T_{\beta^*}) + \beta (D - I_Y(\T_{\beta^*}))\).
Moreover, since the primal problem is assumed to have a feasible solution, the dual function is bounded above. 
Therefore, as \(d(\beta)\) is a linear function on \([\beta_{\textrm{max}},\infty)\) with slope \(D - I_Y(\T_{\beta^*})\) that is bounded above, it must be that \((D - I_Y(\T_{\beta^*})) \leq 0\).
Consequently, \(d(\beta)\) is monotone decreasing on \([\beta_{\textrm{max}},\infty)\).
As a result, we have that for \(\beta_{\textrm{max}} \leq \beta^*\), \(d(\beta_{\textrm{max}}) \geq d(\beta^*)\), and thus \(d(\beta_{\textrm{max}}) = d(\beta^*)\).
Since \(\beta_{\textrm{max}} \in \betaPhaseTransitionSet\), it follows that there exists a tree phase transition \(\hat \beta = \beta_{\textrm{max}}\) such that \(d(\hat\beta) = d(\beta^*)\).

\vspace{8pt}

\noindent
Case III: \(\beta_{\textrm{min}}\leq\beta^*\leq\beta_{\textrm{max}}\).
In this case, there exists \(\beta_{\textrm{lb}} = \max\{\beta:\beta\leq \beta^*,~\beta\in\betaPhaseTransitionSet\}\) and \(\beta_{\textrm{ub}} = \min\{\beta:\beta^*\leq\beta,~\beta\in\betaPhaseTransitionSet\}\).
Notice that \(\beta_{\textrm{lb}}\) and \(\beta_{\textrm{ub}}\) are consecutive tree phase transitions satisfying \(\beta_{\textrm{lb}} \leq \beta^* \leq \beta_{\textrm{ub}}\).
Consequently, from~\eqref{eq:dualIBtreeFunctionWithDeltaInfo} and the definition of tree phase transition, it follows that, for any \(\beta\in [\beta_{\textrm{lb}},\beta_{\textrm{ub}}]\), the dual function is given by \(d(\beta) = I_X(\T_{\beta^*})+\beta (D - I_Y(\T_{\beta^*}))\).
There are now two possibilities: \((D - I_Y(\T_{\beta^*})) \geq 0\) or \((D - I_Y(\T_{\beta^*})) \leq 0\).
In the first scenario, \(d(\beta)\) is monotone increasing on \([\beta_{\textrm{lb}},\beta_{\textrm{ub}}]\).
Therefore \(d(\beta_{\textrm{ub}}) \geq d(\beta^*)\) and thus \(d(\beta_{\textrm{ub}}) = d(\beta^*)\).
In the second scenario, \(d(\beta)\) is monotone decreasing, and so \(d(\beta_{\textrm{lb}}) \geq d(\beta^*)\) leading to \(d(\beta_{\textrm{lb}}) = d(\beta^*)\).
Since both \(\beta_{\textrm{lb}}\) and \(\beta_{\textrm{ub}}\) are tree phase transitions, it follows that there exists a tree phase transition \(\hat \beta\) such that \(d(\hat \beta) = d(\beta^*)\).
\hfill \(\qedsymbol\)

\section{Proof of Theorem~\ref{thm:charaterizeStrongDualityIBtree}}\label{app:charaterizeStrongDualityIBtreeProof}

\noindent 
We first show that, if \(D = [Q_{Y,\tRoot}]_j\) for some \(j\) then strong duality holds in the IB tree problem.
Let \(D = [Q_{Y,\tRoot}]_j\) for some \(j\).
Then, choose \(\beta\geq0\) according to 
\begin{equation*}
    \beta = 
    \begin{cases}
        [\betaPhaseTransitionSetNode{\tRoot}]_{j+1}, & \text{ if } j < \lvert \betaPhaseTransitionSet \rvert, \\
        \max(\betaPhaseTransitionSetNode{\tRoot}) + \varepsilon, & \text{ otherwise},
    \end{cases}
\end{equation*}
where \(\varepsilon > 0\).
Next, let \(\T_\beta \in \T^\Q\) be a solution to the problem
\begin{equation*}
    \min\{I_X(\T) - \beta I_Y(\T) : \T\in\T^\Q\},
\end{equation*}
obtained from the Q-tree search algorithm.
Now, since \([\betaPhaseTransitionSetNode{\tRoot}]_j < \beta \leq [\betaPhaseTransitionSetNode{\tRoot}]_{j+1}\), it follows, by construction of \(Q_{Y,\tRoot}\) that the solution \(\T_{\beta} \in\T^\Q\) satisfies \(I_Y(\T_{\beta}) = [Q_{Y,\tRoot}]_j\).
Notice that \(\T_\beta\) is primal feasible since \(I_Y(\T_{\beta}) \geq D\).
Moreover, observe that \(\beta \geq 0\) and because \(\T_{\beta}\in\T^\Q\) is primal feasible, we have 
\begin{equation*}
    I_X(\T_{\beta}) + \beta (D - I_Y(\T_\beta)) \leq d(\beta^*) \leq I_X(\T^*) \leq I_X(\T_{\beta}),
\end{equation*}
where \(\T^* \in \T^\Q\) is an optimal solution to the primal problem.
However, since \(D - I_Y(\T_\beta) = 0\), we have that \(I_X(\T_\beta) + \beta(D - I_Y(\T_{\beta}) = I_X(\T_\beta)\), which means
\begin{align*}
    I_X(\T_\beta) = I_X(\T_\beta) + \beta(D-I_Y(\T_\beta))&\leq d(\beta^*),\\
    &\leq I_X(\T^*) \leq I_X(\T_\beta).
\end{align*}
Therefore, \(\T_\beta\) is primal optimal, \(\beta\) is dual optimal and \(I_X(\T_\beta) = I_X(\T^*) = d(\beta^*)\); that is, strong duality holds.

\vspace{5pt}

\noindent
Next, we show that if strong duality holds in the IB tree problem, then \(D = [Q_{Y,\tRoot}]_j\) for some \(j\).
To this end, since strong duality holds, there exists a primal solution \(\T^*\in\T^\Q\) and a dual solution \(\beta^* \geq 0\) such that \(I_X(\T^*) = d(\beta^*)\).
Consequently, we have from strong duality, Lemma~\ref{lem:suffDualAndPTsOPT}, relation~\eqref{eq:dualFunctionInTermsOfPTalgVars}, and the properties of minimum that
\begin{align*}
    I_X(\T^*) = d(\beta^*) &= d(\hat \beta), \\
    &= [Q_{X,\tRoot}]_j + \hat \beta (D - [Q_{Y,\tRoot}]_j),\\
    &\leq I_X(\T^*) + \hat \beta (D - I_Y(\T^*)),\\
    &\leq I_X(\T^*),
\end{align*}
where the final inequality follows from the fact that the tree phase transition \(\hat \beta \geq 0\) and, since \(\T^*\) is primal feasible, \(D - I_Y(\T^*) \leq 0\).
The above implies
\begin{align*}
    I_X(\T^*) &= [Q_{X,\tRoot}]_j + \hat \beta (D - [Q_{Y,\tRoot}]_j),\\
    &= I_X(\T^*) + \hat \beta (D - I_Y(\T^*)),
\end{align*}
and thus, \(I_X(\T^*) = [Q_{X,\tRoot}]_j\)  and \(I_Y(\T^*) = [Q_{Y,\tRoot}]_j\) for some \(j\).
Moreover, for \(I_X(\T^*) = I_X(\T^*) + \hat \beta (D - I_Y(\T^*))\), we must have
\begin{equation*}
    \hat \beta (D - I_Y(\T^*)) = \hat \beta (D - [Q_{Y,\tRoot}]_j) = 0.
\end{equation*}
However, since \(\hat \beta\) is a tree phase transition, it follows that \(\hat\beta \geq  \min(\betaPhaseTransitionSet) = \betacrq(\tRoot) > 0\).
Consequently, strong duality implies \(D - [Q_{Y,\tRoot}]_j = 0\), or equivalently, that \(D = [Q_{Y,\tRoot}]_j\) for some \(j\). \hfill \(\qedsymbol\)

\section{Proof of Proposition~\ref{prop:TUMIBtrees}}\label{app:proofOfTUMProp}

Let \(J \subseteq \{1,\ldots,n\}\) be any collection of columns of \(A\).
Notice that from~\eqref{eq:fullILPcons2}, each row of \(A\) has a column containing a \(1\), a column containing a \(-1\), with the rest of the entries in the row being \(0\).
Now consider the partition of \(J\) given by \(J_1 = J\) and \(J_2 = \varnothing\).
For this partition, it follows that \(\sum_{j\in J_1} [A]_{ij}\) equals either \(-1\), \(0\), or \(1\) for every row \(i\).
As a result, we have \(\lvert \sum_{j\in J_1} [A]_{ij} \rvert \leq 1\).
Therefore, for any collection \(J \subseteq \{1,\ldots,n\}\) of columns of \(A\), there exists a partition \(J_1 = J\) and \(J_2 = \varnothing\) such that \(\lvert \sum_{j\in J_1}[A]_{ij} - \sum_{j\in J_2} [A]_{ij}\rvert \leq 1\) for every row \(i\) of \(A\).
Invoking Theorem~\ref{thm:necSufTUMConditions} we conclude that \(A\) is totally unimodular. \hfill \(\qedsymbol\)

\end{appendices}

\bibliographystyle{IEEEtran}


\end{document}